\documentclass[11pt]{article}

\usepackage{microtype}
\usepackage{nicefrac}
\usepackage{paralist}
\usepackage{fullpage}
\usepackage{color,xspace}
\usepackage{amsmath,amssymb,amsthm,enumerate}
\usepackage{thm-restate}
\usepackage{graphicx}
\usepackage{subcaption}
\usepackage{url}
\usepackage{hyperref}
\hypersetup{colorlinks=true}
\hypersetup{linkcolor=black,anchorcolor=black,citecolor=black,urlcolor=black}
\usepackage[algo2e,ruled]{algorithm2e}
\usepackage{algorithm,  algorithmic}
\usepackage{verbatim}
\usepackage{bm}

\newtheorem{theorem}{Theorem}[section]
\newtheorem{definition}[theorem]{Definition}
\newtheorem{example}{Example}
\newtheorem{prop}[theorem]{Proposition}
\newtheorem{observation}[theorem]{Observation}
\newtheorem{lemma}[theorem]{Lemma}
\newtheorem{corollary}[theorem]{Corollary}
\newtheorem{claim}[theorem]{Claim}

\newtheorem{remark}[theorem]{Remark}

 \newcommand{\alert}[1]{\textbf{\color{red}
[#1]}\marginpar{\textbf{\color{red}**}}\typeout{ALERT:
\the\inputlineno: #1}}

\newenvironment{proofof}[1]{

\noindent{\it Proof of {#1}:}}{\hfill\qed}

\usepackage[colorinlistoftodos,textsize=tiny,textwidth=2cm,color=red!25!white,obeyFinal]{todonotes}

\newcommand{\ignore}[1]{}

\newcommand{\ar}{\Phi}
\DeclareMathOperator{\OPT}{OPT}
\DeclareMathOperator{\ALG}{ALG}

\DeclareMathOperator{\diam}{diam}
\DeclareMathOperator{\polylog}{polylog}

\newcommand{\Ghat}{\widehat{G}}

\newcommand{\Dcal}{\mathcal{D}}

\newcommand{\dhat}{\widehat{d}}
\newcommand{\Acal}{\mathcal{A}}
\newcommand{\Acalhat}{\widehat{\mathcal{A}}}
\newcommand{\fmax}{{\rho}_{\operatorname{max}}}
\newcommand{\fmin}{{\rho}_{\operatorname{min}}}
\newcommand{\Ehat}{\widehat{G}}
\newcommand{\cost}{\operatorname{cost}}

\newcommand{\poly}{\operatorname{poly}}
\newcommand{\R}{\mathcal{R}}
\newcommand{\Scal}{\mathcal{S}}
\newcommand{\Fcal}{\mathcal{F}}
\newcommand{\Zcal}{\mathcal{Z}}
\newcommand{\Rhat}{\widehat{R}}
\newcommand{\Chat}{\widehat{C}}
\newcommand{\Mcal}{\mathcal{M}}
\newcommand{\Talg}{\textsf{Alg}^\mathcal{T}}
\newcommand{\Malg}{\textsf{Alg}^\Mcal}
\newcommand{\Tembed}{\textsf{Embed}^\mathcal{T}}
\newcommand{\Membed}{\textsf{Embed}^\Mcal}
\newcommand{\Balg}{\textsf{BaseAlg}}

\newcommand{\alg}{\textsf{Alg}}
\newcommand{\combalg}{\textsf{CombineAlg}}

\newcommand{\event}{\mathcal{E}}

\allowdisplaybreaks

\begin{document}

\title{Online Probabilistic Metric Embedding: A General Framework for Bypassing Inherent Bounds\thanks{A preliminary version of this work appeared in SODA 2020~\cite{proceedings}. The proof of Lemma 3.4 in the preliminary version was later found to be flawed. This paper presents a corrected proof (Lemma~\ref{lem:base-HST}).}}
  
\author{
Yair Bartal\thanks{Hebrew University of Jerusalem, Israel. \href{mailto:yair@cs.huji.ac.il}{\texttt{yair@cs.huji.ac.il}} Supported in part by ISF grant 1817/17.}
\and
Ora N. Fandina\thanks{Hebrew University of Jerusalem, Israel. \href{mailto:fandina@gmail.com}{\texttt{fandina@cs.huji.ac.il}} Supported in part by ISF grant 1817/17 and by the Leibnitz Foundation.}
\and
Seeun William Umboh\thanks{School of Computing and Information Systems, The University of Melbourne, Australia. \href{mailto:william.umboh@unimelb.edu.au}{\texttt{william.umboh@unimelb.edu.au}} Supported in part by NWO grant 639.022.211 and ISF grant 1817/17. Part of this work was done while at Eindhoven University of Technology and The University of Sydney, and while visiting the Hebrew University of Jerusalem.}
}

\date{}
\begin{titlepage}
\def\thepage{}
\thispagestyle{empty}
\maketitle


\begin{abstract}
  
Probabilistic metric embedding into trees is a powerful technique for designing online algorithms. The standard approach is to embed the entire underlying metric into a tree metric and then solve the problem on the latter. The overhead in the competitive ratio depends on the expected distortion of the embedding, which is logarithmic in $n$, the size of the underlying metric. For many online applications, such as online network design problems, it is natural to ask if it is possible to construct such embeddings in an online fashion such that the distortion would be a polylogarithmic function of $k$, the number of terminals.

Our first main contribution is answering this question negatively, exhibiting a \emph{lower bound} of $\tilde{\Omega}(\log k \log \ar)$, where $\ar$ is the aspect ratio of the set of terminals, showing that a simple modification of the probabilistic embedding into trees of Bartal (FOCS 1996), which has expected distortion of $O(\log k \log \ar)$, is \emph{nearly-tight}. Unfortunately, this may result in a very bad dependence in terms of $k$, namely, a power of $k$. 

Our second main contribution is a general framework for bypassing this limitation. We show that for a large class of online problems this online probabilistic embedding can still be used to devise an algorithm with $O(\min\{\log k\log (k\lambda),\log^3 k\})$ overhead in the competitive ratio, where $k$ is the current number of terminals, and $\lambda$ is a measure of subadditivity of the cost function, which is at most $r$, the current number of requests. In particular, this implies the first algorithms with competitive ratio $\polylog(k)$ for online \emph{subadditive network design} (\emph{buy-at-bulk network design} being a special case), and $\polylog(k,r)$ for online \emph{group Steiner forest}.

 \end{abstract}
\end{titlepage}

\section{Introduction}

Low-distortion metric embeddings play an important role in the design of algorithms. Many applications are based on reducing the problem on arbitrary metrics to instances defined over a simple metric, where we may be able to devise efficient approximation or online algorithms, while incurring only a small loss in quality of the solution depending on the distortion of the embedding.

One of the most successful methods has been probabilistic embedding into ultrametrics \cite{Bar96} where an arbitrary metric space is embedded probabilistically into special tree metrics (HSTs, or ultrametrics). The main theorem provides a tight $O(\log n)$ bound on the distortion \cite{Bartal04, FRT04}, where $n$ is the size of the metric. This result has been used in a plethora of applications in various areas, including clustering \cite{BartalCR01}, metric labeling \cite{KleinbergT02}, network design \cite{AwerbuchA97,GargKR98}, linear arrangement and spreading metrics \cite{Bartal04}, among others. See \cite{I01} for more applications.  

The goal of this paper is to deepen the study of metric embedding and its applications in the online setting. An \emph{online algorithm} receives requests one by one and needs to satisfy each request immediately without knowing future requests. Previously, the dominant approach in applying metric embeddings to online algorithms is to use an offline embedding of the entire underlying metric $(V, d_V)$ into a tree metric or HST using the probabilistic embeddings of \cite{Bartal04, FRT04} and then solving the problem there. 
This approach has been particularly useful for online algorithms, with many applications including subadditive network design \cite{AwerbuchA97}, group Steiner tree/forest \cite{AlonAABN06, NaorPS11}, metrical task system and $k$-server problems \cite{BartalBBT97,BansalBMN15,BCLLM18}, reordering buffers \cite{EnglertRW10}, and file migration \cite{Bar96}, among others.

There are two main disadvantages to this approach: The first is that it requires upfront knowledge of the underlying metric; The second, perhaps more crucial, is the inherent $O(\log n)$ overhead incurred in the resulting competitive ratio bound, adding a dependence on $n$, which may be unrelated to the natural parameters of the problem (e.g., in terms of the problem's lower bounds). This first issue, i.e. the necessity of knowing the underlying metric, is due to the use of the offline embeddings of  \cite{Bartal04, FRT04}. This has motivated several researchers, including Indyk et al.~\cite{IMS10}, to initiate the study of \emph{online embeddings}. In this setting, the embedding algorithm is given the points of the metric one by one and needs to produce consistent metric embeddings at each stage. They observed that the probabilistic embedding method of Bartal~\cite{Bar96} can be quite simply implemented in an online fashion to achieve the distortion bound of $O(\log \hat{k}\log \hat{\Phi})$, where $\hat{k}$ is the \emph{final} number of requested points, and $\hat{\Phi}$ is the aspect ratio of the underlying metric. However, while this alternative approach (partially) overcomes the first issue, it fails to address the second: First, it still requires pre-knowledge of $\hat{k}$---the final number of requested points---and second, it introduces an additional dependence on another parameter, the aspect ratio, which may in general be exponential in $\hat{k}$, and by itself is often also not a natural parameter of the problem.

This paper makes two main contributions: First, we show that the aspect ratio is a \emph{necessary} parameter for any strictly online embedding, more specifically showing that the above bound cannot be much improved except for removing the necessity to know the parameters in advance. The second main contribution is providing a general framework for overcoming this obstacle in order to use such embeddings to obtain general bounds for online problems based on algorithms for trees with an overhead that is polylogarithmic in $k$, the number of terminal points which have appeared in requests up to the current stage. We expect further extensions and applications of this approach to arise in the near future.

\paragraph{Online embedding model.}
An online embedding algorithm receives a sequence of points (called \emph{terminals}) from an underlying metric space $(V,d_V)$ one by one; when a point arrives, the algorithm embeds it into a destination metric space $(M, d_M)$ while committing to previous decisions. 
If the underlying metric space $(V, d_V)$ is given to the algorithm upfront, we say that online embedding is {\it non-metric-oblivious}, otherwise the embedding is {\it metric-oblivious}.

More formally, an online embedding is defined as follows.
\begin{definition}[Online embedding]\label{def_online_emb}
An \emph{online embedding} $f$ of a sequence of terminal points $x_1, \ldots, x_k$ from a metric space $(V,d_V)$ into a metric space $(M, d_M)$ is a sequence of embeddings $f_1, \ldots, f_k$ such that for each $ 1\leq i \leq k$:
\begin{inparaenum}[\bfseries (1)]
\item $f_i$ is a non-contractive embedding of $x_1, \ldots, x_i$ into $M$,
\item $f_i$ \emph{extends} the previous embedding $f_{i-1}$, i.e. for each $x \in \{x_1, \ldots, x_{i-1}\}$, $f_i(x)=f_{i-1}(x)$. 
\end{inparaenum}
The \emph{distortion} of the online embedding $f$ is the distortion of the final embedding $f_k$. A \emph{probabilistic online embedding} is a probability distribution over online embeddings.
\end{definition}

\noindent Note that the results of \cite{Bartal04, FRT04} imply that there exists a probabilistic online non-metric-oblivious embedding into hierarchically separated trees with expected distortion $O(\log n)$, where $n = |V|$. A \emph{hierarchically separated tree} (HST) is a rooted tree $T$ where each edge has a non-negative length and the edge lengths decrease geometrically along every root-to-leaf path.

For our upper bounds, we will actually construct a metric-oblivious embedding, while we will prove lower bounds for non-metric-oblivious embeddings. 

\paragraph{Fully extendable embeddings.}
A key application of online embeddings is to reduce an online problem defined over an arbitrary metric $(V,d_V)$ to one defined over a ``simpler'' metric $(M,d_M)$. For this to work, we should be able to translate a  feasible solution obtained on $(M,d_M)$, which in general may use Steiner points (points in $M$ that are not images of the terminals) to a feasible solution in $(V,d_V)$. Moreover, for the analysis of the reduction, we should be able to translate an optimal solution in $(V,d_V)$ into a feasible solution in $(M,d_M)$. Here, it may definitely be the case that Steiner points in $V$, i.e. non-terminal points, are used in an optimal solution in $V$. We therefore further require the online embedding to be \emph{fully extendable}, that is both $f$ and $f^{-1}$ must have appropriate metric extensions.

\begin{restatable}[Fully extendable embeddings]{definition}{fullyextendable}\label{extendability}
     Let $f = (f_1, \ldots, f_k)$ be an online embedding of a sequence of terminals $x_1, \ldots, x_k \in  V$ into $M$. For all $ 1\leq i \leq k$, let $\alpha_i$ denote the bound on distortion of $f_i$.  We say that $f$ is \emph{fully extendable} to $V$ with respect to $\{M_i\}_{i \leq k}$, $M_1 \subseteq M_2 \ldots \subseteq M_k \subseteq M$,  if for all $i$ there exist $F_i \colon V \to M$ and $H_i \colon M_i \to V$  such that\footnote{The inverse of an injective function should be interpreted for its restriction to its image.}:
 \begin{inparaenum}[\bfseries (i)]
 \item\label{ext:1} $F_i$ and $H_i$ depend only on terminals $\{x_1, \ldots, x_i\}$;
 \item\label{ext:2} $F_i$ extends $f_i$ and for all $u \neq v \in V$, ${d_M(F_i(u), F_i(v))} \leq \alpha_i \cdot {d_V(u,v)}$;
 \item\label{ext:3} $H_i$ extends $f_i^{-1}$ and for all ${\hat u} \neq {\hat v} \in M_i$, $d_V(H_i({\hat u}), H_i({\hat v})) \leq  d_{M}({\hat u}, {\hat v})$; 
 \item\label{ext:4} $H_{i+1}$ extends $H_{i}$.
 \end{inparaenum}
 A probabilistic online embedding is fully extendable if for each embedding $f$ in its support, there exist extensions $F_i$ and $H_i$ such that properties (\ref{ext:1}), (\ref{ext:3}), and (\ref{ext:4}) hold for each embedding, and property (\ref{ext:2}) holds in expectation, i.e. for all $u \neq v \in V$, $E_{F_i}[d_M(F_i(u),F_i(v))] \leq \alpha_i \cdot d_V(u,v)$, where $\alpha_i$ is the expected distortion of the probabilistic embedding $f_i$. 
 \end{restatable}
 \noindent Roughly speaking, one should think of the non-terminal points in $M_i$ as Steiner points that are useful in defining $M_i$. For example, if $f_i$ embeds the terminals into the leaves of a tree $T_i$, then $M_i = T_i$.

\subsection{Our Results}
We begin by showing that a somewhat more delicate application of the approach of Bartal~\cite{Bar96} 
yields a probabilistic online embedding into HSTs with expected distortion $O(\log k \log \ar)$ which does not require knowledge of the underlying metric space, i.e., the embedding is \emph{metric-oblivious}.

\begin{theorem}
  \label{thm:dist-ub}
  For any metric space $(V, d_V)$ and sequence of terminals $x_1, \ldots, x_k \in V$, there exists a probabilistic online embedding into HSTs with expected distortion 
$O(\log k \log \Phi)$,
where $\Phi$ is the aspect ratio of the terminals.\footnote{Note that this can be made into $O(k \log k)$. Essentially, we only need to maintain a probabilistic partition for each ``relevant'' distance scale and there can be at most $\min\{k, \log \Phi\}$ of these.}
\end{theorem}

Our first main result is a lower bound on the distortion of every probabilistic online embedding into trees, showing that the above bound is nearly tight. We also prove a better lower bound for probabilistic online embeddings into HSTs. We note that the lower bounds hold for metric-non-oblivious online embeddings.

\begin{theorem}
  \label{thm:dist-lb}
  There exists an infinite family of metric spaces $\{(V_\ell, d_{V_\ell})\}$ and terminal sequences $\{\sigma_\ell\}$ such that every probabilistic online embedding of $\sigma_\ell$ into trees has expected distortion 
$\Omega\left(\log k \log \Phi_\ell/(\log \log k + \log\log \Phi_\ell)\right)$,
where $k = |\sigma_\ell|$ and $\Phi_\ell$ is the aspect ratio of $(V_\ell, d_{V_\ell})$ which can be as large as $ 2^{k^{(1-\delta)/2}}$ for any fixed $\delta \in (0,1)$.
\end{theorem}
\begin{theorem}
  \label{thm:dist-HST-lb}
  There exists an infinite family of metric spaces $\{(V_\ell, d_{V_\ell})\}$ and terminal sequences $\{\sigma_\ell\}$ such that every probabilistic online embedding of $\sigma_\ell$ into HSTs has expected distortion at least $\Omega\left(\log k \log \Phi_\ell/\log \log k\right)$,
where $k = |\sigma_\ell|$ and $\Phi_\ell$ is the aspect ratio of $(V_\ell, d_{V_\ell})$ which can be as large as $2^{k^{1-\delta}}$ for any fixed $\delta \in (0,1)$.
\end{theorem}
\noindent This is almost-tight and also implies that as a function of $k$, a dependence of some power of $k$ is required.

Our second main result is that despite the lower bound of Theorem~\ref{thm:dist-HST-lb}, which forces a dependence of $\Phi$ on the embedding distortion, for a broad class of online problems, it is still possible to use the HST embedding of Theorem~\ref{thm:dist-ub} in a more clever manner such that the overhead of using the embedding is only $O(\min\{\log k \log (kr), \log^3 k\})$, where $r$ is the number of requests. This is much less than the distortion of the embedding, effectively bypassing the distortion lower bound of Theorem~\ref{thm:dist-HST-lb}. We call this class of problems \emph{abstract network design}. It captures problems on metric spaces that are amenable to the usual tree embedding approach, including well-known network design problems such as group Steiner forest and buy-at-bulk network design, and other online problems involving metrics such as $s$-server\footnote{This is usually called $k$-server but we use $s$ to denote the number of servers to avoid confusion with our use of $k$.}, reordering buffer, and constrained file migration.

\paragraph{Abstract network design.}
The basic idea behind the definition is to capture problems on graphs that are amenable to the usual tree-embedding-based approach, and whose feasible solutions can be defined in terms of a subset of vertices, called terminals.

In an instance of abstract network design, the algorithm is given a connected graph $G = (V,E)$ with edge lengths $d : E \rightarrow \mathbb{R}_{\geq 0}$. At each time step $i$, the algorithm is given a request that consists of a set of terminals $Z_i \subseteq V$. Let $\mathcal{Z}_i=\cup_{j \leq i} Z_j$, the set of terminals seen so far. The algorithm has to respond with a response $\mathcal{R}_i=(R_i, C_i)$, where $R_i \subseteq E$ is a subgraph of $G$, and a \emph{connectivity list} $C_i$ which is an ordered subset of terminal pairs from ${\mathcal{Z}_i \choose 2}$. A solution to the first $i$ requests is a sequence of responses $\mathcal{S}_i:= (\mathcal{R}_1, \mathcal{R}_2, \ldots, \mathcal{R}_i)$. The algorithm is also given at each time step a \emph{feasibility function} $\mathcal{F}_i=\mathcal{F}_i[Z_1, Z_2, \ldots, Z_i]$ mapping a sequence of connectivity lists $(C_1, C_2, \ldots, C_i)$ to $\{0,1\}$. The solution $\mathcal{S}_i$ is \emph{feasible} iff $\mathcal{F}_i(C_1, C_2, \ldots, C_i)=1$ and every pair in $C_j$ is connected in $R_j$ for each $j \leq i$.  A valid algorithm must maintain feasible solutions $\mathcal{S}_1, \mathcal{S}_2, \ldots \mathcal{S}_r$.

A feasibility function is called \emph{memoryless} if it obeys that whenever $\mathcal{F}_i(C_1, C_2, \ldots, C_i)=1$, then for all $(C_1', C_2', \ldots, C_{i-1}')$ such that $\mathcal{F}_{i-1}(C_1', C_2', \ldots, C_{i-1}')=1$, it holds that $\mathcal{F}_i(C_1', C_2', \ldots, C_{i-1}',C_i)=1$. E.g., most standard network design problems, such as Steiner tree, possess a memoryless feasibility function.

Next, we specify how the cost of a solution $\Scal$ is determined. To include a wide class of problems, we use the general framework of subadditive functions: a set function $f$ is \emph{subadditive} if for any pair of sets $A,B$, it holds $f(A \cup B) \leq f(A) + f(B)$. In particular, we define the cost of using an edge $e$ to be its length $d(e)$ times a subadditive function of the subset $I$ of time steps $i$ in which $e \in R_i$. At each time step $i$, the algorithm is given oracle access to a \emph{load function} $\rho_i: 2^{\{1, \ldots, i\}} \to \mathbb{R}_{\geq 0}$ that is  subadditive, monotone-increasing and satisfies $\rho_i(I) = 0$ if and only if $I = \emptyset$. We also require that $\rho_i$ extends $\rho_{i-1}$ for each $i$. To simplify notation, for a set $I$ of time steps, we write $\rho(I) = \rho_j(I)$ where $j$ is the latest time step in $I$. The cost of a solution $\mathcal{S}_i$ is $\cost(\mathcal{S}_i) = \sum_{e \in E} d(e) \rho(\{ j \leq i : e \in R_j\})$. 
The goal is to minimize the cost. When the algorithm is randomized, we measure the expected cost of its solution.
We denote the number of requests in the online problem by $r$ and write $\Zcal = \Zcal_r$, and we use $k$ to denote the total number of terminal points $|\mathcal{Z}_r|$. We also define the parameter $\lambda_\rho = \min\{\max_I \rho(I)/\min_{I \neq \emptyset}\rho(I), r\}$ which is, intuitively,  a measure of the level of subadditivity of the function.

It is often convenient to assume w.l.o.g. (the proof is in Appendix \ref{metric_completion}) that $G$ is a complete graph, with edge lengths satisfying triangle inequality.

One can also consider \emph{metric-oblivious} abstract network design, where the input graph $G$ is not given to the algorithm in advance. In this case we will assume that $G$ is a metric graph and at each time step $i$, subset $Z_i$ together with the pairwise distances of $\mathcal{Z}_i$ are revealed to the algorithm. 
All our results generalize to the metric-oblivious version. For simplicity of the exposition, in the rest of the paper we assume that the input graph $G$ is given to the algorithm.

For the sake of concreteness, we demonstrate how to express four of our applications in our framework. The first three problems share the same load function $\rho$ which satisfies $\rho(I) = 1$ when $I \neq \emptyset$ and $\rho(\emptyset) = 0$; note that $\lambda_\rho = 1$.

\begin{example}[Steiner Forest]
  In the online Steiner Forest problem, a sequence of vertex pairs $(s_i, t_i)$ arrives online, and the goal is to maintain a subgraph $G'$ such that each $(s_i,t_i)$ that has arrived so far is connected in $G'$. The cost of a solution $G'$ is the total length of its edges.
  
  In the Abstract Network Design formulation, the terminal set of the $i$-th request consists of $s_i$ and $t_i$. The feasibility function $\mathcal{F}_i$ satisfies $\mathcal{F}_i(C_1, \ldots, C_i) = 1$ if and only if $C_j = \{(s_j,t_j)\}$ for every $1 \leq j \leq i$. 
\end{example}

\begin{example}[Constrained Forest]
  In the online Constrained Forest problem \cite{GoemansW95}, each request consists of a set of terminals $Z_i$ and a \emph{cut requirement function} $g_i : 2^{Z_i} \to \{0,1\}$ that is \emph{proper}: $g_i(\emptyset) = g_i(Z_i) = 0$, $g_i(X) = g_i(Z_i - X)$ for all $X \subseteq Z_i$ and $g_i(X \cup Y) \leq \max\{g_i(X), g_i(Y)\}$ for all disjoint sets $X, Y \subseteq Z_i$. The goal is to maintain a subgraph $G'$ that satisfies each request $(Z_i, g_i)$ seen so far, i.e. for every vertex subset $S \subseteq V$ such that $g_i(S \cap Z_i) = 1$, there is at least one edge in $G'$ that leaves the set $S$. The cost of $G'$ is the total length of its edges.

  In the Abstract Network Design formulation, the terminal set of the $i$-th request is $Z_i$. The feasibility function $\mathcal{F}_i$ satisfies $\mathcal{F}_i(C_1, \ldots, C_i) = 1$ if and only if for every $1 \leq j \leq i$, for every terminal subset $S \subseteq Z_j$ such that $g_j(S \cap Z_j) = 1$, there is at least one edge in $C_j$ that leaves the set $S$. 
\end{example}

\begin{example}[Group Steiner Forest]
    In the online Group Steiner Forest problem, each request consists of a pair of vertex sets $(S_i, T_i)$. The goal is to maintain a subgraph $G'$ such that for each request $(S_i,T_i)$ seen so far, we have that $G'$ connects some vertex in $S_i$ to some vertex in $T_i$. The cost of $G'$ is the total length of its edges.

      In the Abstract Network Design formulation, the terminal set of the $i$-th request is $S_i \cup T_i$. The feasibility function $\mathcal{F}_i$ satisfies $\mathcal{F}_i(C_1, \ldots, C_i) = 1$ if and only if for every $1 \leq j \leq i$, we have that $C_i = \{(s_i,t_i)\}$ for some $s_i \in S_i$ and some $t_i \in T_i$.
    \end{example}

\begin{example}[$s$-server]
  In the $s$-server problem, we are given initial locations of $s$ servers. Each request consists of a single vertex $t_i$. Upon receiving a request, it must be served by moving one of the servers from its current location to $t_i$, incurring a cost equal to the distance travelled.

  In the Abstract Network Design formulation, the terminal set of the $i$-th request consists of the set of points on which the first $i$ requests of the $s$-server instance appeared on and the initial locations of the $s$ servers. The feasibility function $\mathcal{F}_i$ satisfies $\mathcal{F}_i(C_1, \ldots, C_i) = 1$ if for each $j \leq i$, $C_j$ is a single pair $(u,v)$ (which corresponds to moving a server from $u$ to a request location $v$) and the sequence of $(C_1, \ldots, C_i)$ corresponds to a valid movement of the servers. The load function $\rho$ is simply the cardinality function since we pay $d(e)$ each time $e$ is traversed; here $\lambda_\rho = r$.  
\end{example}

\begin{example}[Reordering Buffer]
  In the reordering buffer problem, we are given the initial location of a single server and a parameter $b$ called the buffer size. Each request consists of a single vertex $t_i$. A request is served when the server visits it. At any point in time, all but at at most $b$ requests must be served.

  In the Abstract Network Design formulation, the terminal set of the $i$-th request consists of the initial location of the server and the set of all requests thus far. The feasibility function $\mathcal{F}_i$ and the load function $\rho$ are defined as in the $s$-server problem. 
\end{example}

\noindent Note that the feasibility functions for the online network design problems are memoryless. On the other hand, the $s$-server problem is not memoryless since we need to check that $(C_1, \ldots, C_i)$ corresponds to a valid movement of servers.

Our second main result is that for abstract network design problems satisfying some mild condition, it is possible to use the online embedding of Theorem~\ref{thm:dist-ub} such that the overhead due to the embedding is only $\polylog$ in $k$ and $\lambda_\rho \leq r$. The two ingredients that we need is a ``baseline'' algorithm on general metrics with a finite competitive ratio $\beta$ (that can be arbitrarily large, possibly dependent on $r$, $k$ and $n$), and a scheme for combining any two algorithms $A$ and $B$ to give a ``combined'' algorithm $C$ such that for any instance, the cost of $C$ is not much more than the minimum of the cost of $A$ and $B$.  It has been shown by \cite{ABM93}
that a broad class of problems, namely, Metrical Task Systems, admit a min operator. We also show that memoryless abstract network design problems admit a min operator (Corollary~\ref{cor_min_operator} in Appendix~\ref{onlineMin}).

\begin{definition}[Min operator]\label{def:min_operator-inf}
An online problem admits a \emph{min operator} with factor $\eta\geq 1$ if 
for any two deterministic online algorithms $A$ and $B$, there is a deterministic online algorithm $C$ that on every instance of the problem satisfies $\cost(C) \leq \eta \cdot \min\{\cost(A), \cost(B)\}$, where $\cost(\cdot)$ is the cost of the respective algorithm. Moreover, if either $A$ or $B$ is randomized, then $C$ is also randomized and has $E[\cost(C)] \leq \eta \cdot E[\min\{\cost(A),\cost(B)\}]$. If $\eta=O(1)$, we simply say that the problem admits a min operator.
\end{definition}
\noindent Note that the definition for the case when $A$ and $B$ are deterministic implies the definition for the case when at least one of them is randomized since a randomized algorithm is a distribution over deterministic algorithms.

\begin{theorem}
  \label{thm:subadd-ub}
  Consider an abstract network design problem that admits a min operator and that has a baseline algorithm $B$ on general instances with competitive ratio $\beta$. If there exists an algorithm $A$ that is $\alpha$-competitive on instances where the input graph is an HST then there exists a randomized algorithm that, on every instance, has expected competitive ratio 
  $$O(\alpha \cdot \min\{ \log k \cdot \log (k \alpha \lambda_\rho), \log^2 k \cdot \log(k\alpha)\}).$$

	If $A$ and/or $B$ have additive terms in the competitive ratio, then the additive term is $c^{\log^*(\beta)}$, where $c=O(\max\{a, b\})$, and $a$ and $b$ are the additive terms of $A$ and $B$. 
	
\end{theorem}

Assuming there is an $\alpha$-competitive algorithm on HST's implies the existence of a competitive baseline algorithm on general inputs. Particularly, as we show in Theorem \ref{thm:reduction-trees}, the usual tree embedding based approach results in $\beta=O(\alpha\log k \log \Phi)$ competitive algorithm, when the online embedding of Theorem \ref{thm:dist-ub} is used.

Theorem \ref{thm:subadd-ub} states that the additive term of the algorithm for general inputs stems from the additive term in the algorithm on HST instances. 
 We note that for many online algorithms on trees the competitive ratio has no additive term. In particular, this is the case for all online problems we consider in this paper except the $s$-server problem.
However, for the $s$-server problem there is the baseline algorithm with a constant $(2s-1)$-competitive ratio~\cite{KoutsoupiasP95}.

We also note that for 
the abstract network design problems with a memoryless feasibility function, the greedy online algorithm (i.e., at each time step the algorithm outputs the feasible response $R_i$ with the least cost) is $r$-competitive.

\paragraph{Network design problems.}
Our model naturally captures online network design problems and their generalizations to subadditive costs, i.e.~when the load function is subadditive.As mentioned earlier, these problems have memoryless feasibility functions. In Appendix~\ref{onlineMin} (Corollary~\ref{cor_min_operator}), we show that such problems admit a min operator. We now state applications to two network design problems: Constrained Forest, and Group Steiner Forest.

It is easy to see that the Constrained Forest problem can be solved exactly on trees, even with subadditive costs. Thus, we can apply Theorem~\ref{thm:subadd-ub} with $\alpha = 1$ to get:

\begin{corollary}
 There is a randomized $O(\min\{\log k \cdot \log (k\lambda_\rho), \log^3 k\})$-competitive algorithm for the Subadditive Constrained Forest problem.
\end{corollary}

For the Group Steiner Forest problem, we restrict our attention to the load function $\rho$ where $\rho(I) = 0$ if $I = \emptyset$ and $\rho(I) = 1$ otherwise, i.e. we simply want to minimize the total length of edges used; note that $\lambda_\rho = 1$. Naor, Panigrahi and Singh~\cite{NaorPS11} gave an $O(\log^4 k \log r)$-competitive algorithm on HSTs, which yields an $O(\log n \log^4 k \log r)$-competitive algorithm on general graphs when combined with the usual $O(\log n)$-distortion embedding into HSTs. Applying Theorem~\ref{thm:subadd-ub} with $\lambda_\rho = 1$ and $\alpha = \log^4 k \log r$ gives us:

\begin{corollary}
  There is a randomized $\tilde{O}(\log^6 k \cdot \log r)$-competitive algorithm for the Group Steiner Forest problem. 
\end{corollary}
\noindent Even though our techniques apply to any subadditive load function, there is no known algorithm that can handle general subadditive costs with a polylogarithmic competitive ratio, even on trees.

\paragraph{Other online metric problems.} We consider several more applications of Theorem \ref{thm:subadd-ub}: the constrained file migration~\cite{BartalFR92, AwerbuchBF93, BartalDP96}, $s$-server~\cite{ManasseMS88}, and reordering buffer~\cite{RackeSW02} problems. As all these problems are Metrical Task Systems, they admit a min operator (\cite{ABM93}), implying the applicability of the theorem.

\begin{corollary}\label{cor_paging}
There is a randomized algorithm for the constrained file migration problem with competitive ratio $O(\log^4 m)$, where $m$ is the sum of the number of different pages requested and the number of different locations in the requests thus far.
\end{corollary}

The corollary follows from applying Theorem~\ref{thm:subadd-ub} with the value of $\alpha$ (the randomized competitive ratio for HSTs) being $O(\log m)$ \cite{Bar96}.

We can assume w.l.o.g. that each node in the network holds only one file. Thus, the request of the node $v$ to access the file $F$ is associated with moving the file from the node $w$ that holds it, to the node $v$.  

\begin{corollary}\label{cor_server}
There is a randomized algorithm for the $s$-server problem with competitive ratio $O(\log^2 s \log r(\log r + \log\log s))$.
\end{corollary}
We can assume w.l.o.g. that all the servers are located at the same location at the beginning (the proof is in Appendix ~\ref{app:server}). Therefore, to conclude the corollary, we apply Theorem~\ref{thm:subadd-ub} with the following parameters: $\alpha=O(\log^2 s)$, by \cite{BCLLM18}, $k=r+1$, and $\lambda_{\rho}=r$. We note that there is an alternative approach to the problem with overhead $O(\log s \log \Phi)$, providing a competitive ratio $O(\log^3 s \log \Phi)$ (\cite{BCLLM18}). Our bounds provide an improvement when $r$ is small.

\begin{corollary}\label{cor_buffer}
 There is a randomized algorithm for the reordering buffer problem with competitive ratio $O(\log b\log r(\log r + \log\log b))$, where $b$ is the buffer size.
\end{corollary}

The corollary follows from applying Theorem ~\ref{thm:subadd-ub} with: $\alpha=O(\log b)$, by \cite{ER17}, $k=r+1=O(r)$ and $\lambda_{\rho}=r$. The existing results hold for the case where the metric is known in advance and the overhead is $O(\log n)$.

\subsection{Our Techniques}
We briefly summarize our main technical contributions here. 

\paragraph{Upper bound for online embedding.}
Our online embedding is based on the construction of Bartal~\cite{Bar96}. The key idea is to use probabilistic partitions with a dynamic padding parameter, which can be bounded by $O(\log k)$ using the analyses of \cite{Bartal04, ABN06}. Achieving dependence on the current number of terminals $k$ rather than the final number of requested points $\hat{k}$ is made possible via a technique of \cite{ABN06}, where the probabilistic partitions are used in the construction with padding parameter polynomial in $k$. We maintain a hierarchical partition of the terminals, one per distance scale, such that for every scale, the partition at that scale is a refinement of the higher-scale partitions. A technical challenge is maintaining the refinement property of the partitions as new distance scales appear over time causing an increase in the aspect ratio $\Phi$. We then show that our construction can be modified to be fully extendable (see Section~\ref{sec:online:emb} for a discussion), which is necessary for its applications.

\paragraph{Lower bounds.}
We use a recursive construction to build the underlying graph. The idea is to use as a ``base graph" a high-girth graph which has many long edge-disjoint paths. The existence of such a graph follows from an application of the probabilistic method. Applying Yao's principle, we use the underlying graph to define a distribution over terminal sequences and vertex pairs such that any deterministic online embedding into a tree has expected distortion $\Omega(\ell\log k)$, where $\ell$ is the number of recursive levels.

\paragraph{Abstract network design.}
To prove Theorem~\ref{thm:subadd-ub}, we augment the usual application of tree embeddings to online problems. Consider an abstract network design problem that admits a min operator and an algorithm for HST instances. Our approach uses the combining scheme to combine the usual tree-embedding-based algorithm and the baseline algorithm. At a high level, the scheme allows us to fall back on the baseline algorithm in the case that the tree solution becomes very expensive due to some graph edges being distorted badly. Essentially, this lets us replace the logarithmic dependency on $\Phi$ with a logarithmic dependency on the competitive ratio of the baseline algorithm, the number of terminals $k$ and $\lambda_\rho$ (Theorem~\ref{thm:combine}). Since the dependency on the competitive ratio of the baseline algorithm decreases exponentially, this idea can be applied repeatedly to completely remove the dependency on the baseline algorithm as well.

\subsection{Other Related Work}
The notion of online embedding was previously considered by \cite{IMS10}, where the authors considered the metric-oblivious model. The only previous work that tries to go beyond the $O(\log n)$-distortion overhead of the usual tree embedding approach are the recent results of Bubeck et al.~\cite{BCLLM18} for the $s$-server problem. They gave a dynamic embedding into HSTs for the $s$-server problem with an overhead of $O(\log s\log \Phi)$, where $\Phi$ is the aspect ratio of the entire underlying metric space.

\paragraph{Previous work on subadditive network design.}
In the offline setting, the only results on general subadditive network design that we are aware of is an $O(\log k)$-approximation by combining the Awerbuch-Azar technique with a simple modification of the tree embedding of \cite{FRT04} due to \cite{GuptaNR10}. (See proof of Lemma~\ref{lem:OPT-embed})
Most of the literature on subadditive network design focuses on the well-studied buy-at-bulk case where the subadditive function $f$ is of the form $f(A) = g(|A|)$ for some concave function $g$. For the buy-at-bulk problem, the best approximation is achieved by the same $O(\log k)$-approximation algorithm for the general subadditive problem and there is a hardness result of $O(\log^{1/4 - \epsilon} k)$~\cite{Andrews04}; in the single-sink setting, where the sinks $t_i$ are equal, many $O(1)$-approximations are known, e.g. \cite{GargKKRSS01, Talwar02,  GuptaKR04, GrandoniI06, GuhaMM09, JothiR09, GrandoniR10, GoelP12}.

In the online setting, the only work on the general subadditive network design problem is due to \cite{GuptaHR06}. Assuming that the underlying metric is known to the algorithm, a derandomization of oblivious network design from \cite{GuptaHR06} gives an $O(\log^2 n)$-competitive algorithm. As in the offline setting, most of the previous work is on the buy-at-bulk problem. For the special cases of the Steiner Tree and Steiner Forest problems, tight deterministic $O(\log k)$-competitive algorithms are known \cite{ImaseW91,BermanC97}. There are also deterministic $O(\log k)$-competitive algorithms for the online rent-or-buy problem \cite{AwerbuchAB04, BartalCI01, Umboh15}.  For the buy-at-bulk problem, the only prior work is on the single-sink case: \cite{GRTU17} gave a deterministic $O(\log k)$-competitive algorithm and also observed that is possible to use online tree embeddings and obtain a randomized $O(\log^2 k)$-competitive algorithm.

\paragraph{Subsequent work.}
Since the initial announcement of this work~\cite{proceedings}, there has been further work on online embeddings and their applications. Newman and Rabinovich~\cite{NewmanR20} considered deterministic online embeddings into normed spaces and into trees. Forster, Goranci and Henzinger~\cite{ForsterGH21} initiated the study of dynamic algorithms for probabilistic tree embeddings. Recently, Bhore, Filtser and T\'{o}th~\cite{BhoreFT24} gave improved bounds on embeddings into $\ell_2$ and generalized Theorem~\ref{thm:dist-ub} by showing an online probabilistic metric embedding into HSTs with expected distortion $O(\mathrm{ddim}\cdot \log \Phi)$ where $\mathrm{ddim}$ is the doubling dimension of the set of points that are being embedded. Note that every $k$-point metric space has doubling dimension $O(\log k)$. In terms of applications, Deryckere and Umboh~\cite{DeryckereU23} applied our online embeddings to the problem of online matching with delays.

\subsection{Organization of the Paper}
Section~\ref{sec:prelim} covers basic notation and terminology used in the rest of the paper. In Section~\ref{sec:extendability}, we discuss how to use a fully extendable online embedding from $(V, d_V)$ to $(M, d_M)$ to (approximately) reduce an abstract network design problem defined over the former to one defined over the latter. Then, we give an overview of our online embedding construction in Section~\ref{sec:online:emb} (details are in Appendix~\ref{online_embedding}). Next, we describe our lower bound construction in Section~\ref{sec:lowerbound}. Finally, we present our online algorithm for abstract network design problems that admit a min operator in Section~\ref{sec:framework}.

\section{Preliminaries}\label{sec:prelim}

\noindent
{\bf Notation and terminology.} The \emph{distortion} of an embedding $f: V \to Y$ is its maximum expansion times maximum contraction, i.e.  $distortion(f)=\max_{u \neq v \in V}\frac{d_Y(f(u),f(v))}{d_V(u,v)} \cdot \max_{u \neq v \in V} \frac{d_V(u,v)}{d_Y((f(u)),(f(v)))}$. We say that the embedding is \emph{non-contractive} if for all $u \neq v \in V$, $d_Y(f(u), f(v)) \geq d_V(u,v)$ and \emph{non-expansive} if for all $u \neq v \in V$, $d_Y(f(u), f(v)) \leq d_V(u,v)$
The \emph{aspect ratio} of $V$ is $\Phi(V)=d_{\max}(V)/d_{\min}(V)$ where $d_{\max}(V)=\max_{u \neq v \in V}d_V(u,v)$ and $d_{\min}(V)=\min_{u \neq v \in V}d_V(u,v)$.\\

\noindent
{\bf Hierarchically separated trees (HSTs).}
HST metrics were defined in \cite{Bar96}: 
\begin{definition} 
A $\mu$-HST metric is the metric defined on the leaves of a rooted tree $T$ with the following properties. Each node $v$ of $T$ has an associated label $\Delta(v) \geq 0$, such that $\Delta(v)=0$ iff $v$ is a leaf, and for any two nodes $u \neq v$, if $v$ is a child of $u$ then $\Delta(v) \leq \Delta(u)/\mu$. The distance between two leaves $u \neq v$ is given by $d_T(u,v)=\Delta(lca(u,v))$, where $lca(u,v)$ is the least common ancestor of $u$ and $v$.\end{definition} 

The following definition is equivalent to the one given above (up to a constant), we use these representations interchangeably through the paper:
\begin{definition}\label{HST_Labels}
A $\mu$-HST metric is the shortest path metric defined on the leaves of a weighted tree $T$ that satisfies the following: (1) the edge weight from any node to each of its children is the same, and (2) the edge weights along any path from the root to any leaf are decreasing by a factor of at least $\mu$.
\end{definition}

\section{Using Online Embeddings for Abstract Network Design}\label{sec:extendability}
In this section, we show how to use a fully extendable online embedding $f$ from $(V, d_V)$ to $(M, d_M)$ to reduce an instance of an abstract network design problem $P$ defined over the former to one defined over the latter, with an overhead equal to the distortion of $f$. Together Later, in Section~\ref{sec:online:emb} and Appendix~\ref{online_embedding}, we give a fully extendable probabilistic online embedding into HSTs (Theorem~\ref{thm:dist-ub}). Combining these together gives us Theorem~\ref{thm:subadd-ub}.

\begin{definition}[Instance induced by embedding]
  Consider an instance of $P$ defined on $(V, d_V)$ with request sequence $Z_1, \ldots, Z_r \subseteq V$, feasibility functions $\Fcal_1, \ldots, \Fcal_r$ and load function $\rho$. Let $f = (f_1, \ldots, f_r)$ be an online embedding of the request sequence into a tree metric $(T, d_T)$ and $T_i = T[\Zcal_i]$. Then, $f$ \emph{induces} the following instance of $P$ on $(T,d_T)$: the sequence of requests of the induced instance is $f(Z_1), \ldots, f(Z_r) \subseteq M$; given a sequence of connectivity lists $C'_i \subseteq {f(\Zcal_i) \choose 2}$, its $i$-th feasibility function $\Fcal'_i$ satisfies $\Fcal'_i(C'_1, \ldots, C'_i) = \Fcal_i(f^{-1}(C'_1), \ldots, f^{-1}(C'_i))$; and it has the same load function $\rho$ as the original instance.
\end{definition}

In Section~\ref{sec:using-trees}, we describe how to use fully extendable online embeddings into tree metrics, which will be needed for the proof of Theorem~\ref{thm:subadd-ub} in Section~\ref{sec:framework}. In Section~\ref{sec:using-general}, we describe how to generalize the approach to work with fully extendable online embeddings into other families of metrics. 

\subsection{Using Online Embeddings into Trees}
\label{sec:using-trees}

For an embedding into a tree metric, it will be useful to have extension $H_i$ of $f^{-1}_i$ to the subtree induced by the terminals seen so far.

\begin{definition}[Fully extendable online tree embedding]\label{def_fully_tree_ext}
  Consider an online embedding $f = (f_1, \ldots, f_k)$ of a sequence of terminal points $x_1, \ldots, x_k$ from $(V, d_V)$ into a tree metric $(T, d_T)$. Let $T_i$ be the subtree of $T$ induced by $f_i(\{x_1, \ldots, x_i\})$. Then, we say that is \emph{fully extendable} if it is fully extendable with respect to $\{T_i\}_{i \leq k}$. For a probabilistic embedding, we require that every tree embedding in its support is fully extendable.

  We say that an online tree embedding algorithm is \emph{fully extendable} if given any input metric space $(V, d_V)$ and any online sequence of terminal points from $(V,d_V)$, the algorithm produces a fully extendable online embedding of the terminal sequence into a tree metric.
\end{definition}

A crucial property of instances on tree metrics is that the algorithm only needs to consider the subtree induced by the set of terminals seen so far. 
\begin{prop}
  \label{prop:tree-containment}
  Let $P$ be an online Abstract Network Design problem. Consider an instance of $P$ on a tree metric $(T, d_T)$ with request sequence $Z_1, \ldots Z_r$, and let $T_i=T[\Zcal_i]$ where $\Zcal_i = Z_1 \cup \ldots \cup Z_i$. Then, for any feasible solution $\Scal = ((R_1, C_1), \ldots, (R_r, C_r))$, there exists a feasible solution $\Scal' = ((R'_1, C_1), \ldots, (R'_r, C_r))$ such that each $R'_i$ is contained in $T_i$ and $\cost(\Scal') \leq \cost(\Scal)$.
\end{prop}

\begin{proof}
  Consider the solution $\Scal'$ with $R'_i = \bigcup_{(u,v) \in C_i} p_T(u,v)$, where $p_T(u,v)$ is the unique path in $T$ between $u$ and $v$. Clearly, $\Scal'$ is a feasible solution. We also have that $R'_i \subseteq T_i \cap R_i$. Thus, since the load function is monotone, we get  $\cost(\Scal') \leq \cost(\Scal)$. 
\end{proof}

We are now ready to prove the following reduction theorem.
\begin{theorem}\label{thm:reduction-trees}
  Let $P$ be an online Abstract Network Design problem. Suppose that there exists an online algorithm $\Talg$ for $P$ over tree metrics with competitive ratio $\alpha$, and that there exists a fully extendable online tree embedding algorithm $\Tembed$ with distortion $\gamma$. Then, there exists an online algorithm $\alg$ for $P$ over arbitrary metrics with competitive ratio $\alpha\gamma$. Moreover, when either $\Tembed$ or $\Talg$ is randomized, then $\alg$ is randomized with expected competitive ratio $\alpha\gamma$.
\end{theorem}

\begin{proof}
  For simplicity, we only prove the case when both $\Tembed$ and $\Talg$ are deterministic; the proof for the randomized case is similar. The algorithm $\alg$ works as follows. For each request $Z_i$ and feasibility function $\Fcal_i$, the algorithm $\alg$ uses the online embedding algorithm $\Tembed$ to compute the embedding $f_i$ of $\Zcal_i$, and it feeds the request $f(Z_i)$ and feasibility function $\Fcal'_i$ to $\Talg$. Let $(R^T_i,C^T_i)$ be the response of $\Talg$; note that $R^T_i \subseteq T_i$ by Proposition~\ref{prop:tree-containment}. Then, the algorithm $\alg$ constructs its response $(R_i, C_i)$ by translating $\Talg$'s response to $V$ using the extension function $H_i$ of $f^{-1}$, i.e. $R_i = H_i(R^T_i)$ and $C_i = f^{-1}_i(C^T_i)$. Observe that $x,y \in T$ are connected in $R^T_i$ if and only if $H_i(x),H_i(y) \in V$ is connected in $R_i$. Together with the fact that $H_i$ extends $f^{-1}_i$, we get that every pair in $C_i$ is connected in $R_i$. By feasibility of $\Talg$, we also have $\Fcal_i(C_1, \ldots, C_i) = 1$. Thus, the solution $\Scal_i = ((R_1, C_1), \ldots, (R_i, C_i))$ maintained by algorithm $\alg$ is feasible.

  Finally, we analyze the competitive ratio of algorithm $\alg$. Let $\Scal^* = ((R^*_1, C^*_1), \ldots, (R^*_r, C^*_r))$ denote an optimal solution for the instance on $V$ and $\OPT$ be its cost. Also, let $\OPT(T)$ denote the cost of an optimal solution on $T$.

  \begin{claim}
    \label{clm:1}
		$\cost(\alg) \leq \cost(\Talg).$
  \end{claim}

  \begin{proof}
      For $u,v \in V$, define $I_{u,v} = \{i : (u,v) \in R_i\}$ and similarly, for $(x,y) \in T$, define $I^T_{x,y} = \{i : (x,y) \in R^T_i\}$. For brevity, we write $H = H_r$. Since $H$ extends $H_i$ for every $i < r$, we have $R_i = H_i(R^T_i) = H(R^T_i)$ and so $(u,v) \in R_i$ if and only if there exists $(x,y) \in R^T_i$ such that $H(\{x,y\}) = \{u,v\}$. Thus, $ I_{u,v} = \bigcup_{(x,y) \in T: H(\{x,y\}) = \{u,v\}} I^T_{x,y}$. 
  Using subadditivity of $\rho$, we upper bound $\cost(\Scal)$ by $\cost(\Scal^T)$ as follows:
  \begin{align*}
    \cost(\Scal)
    &= \sum_{u,v \in V} d_V(u,v)\rho(I_{u,v}) \\
    &\leq \sum_{u,v \in V} d_V(u,v) \sum_{(x,y) \in T : H(\{x,y\}) = \{u,v\}} \rho(I^T_{x,y})\\
    &= \sum_{(x,y) \in T} d_V(H(x),H(y)) \rho(I^T_{x,y}) \\
    &\leq  \sum_{(x,y) \in T} d_T(x,y) \rho(I^T_{x,y}) = \cost(\Scal^T),
  \end{align*}
  where the last inequality is due to the non-expansiveness of $H$.
  \end{proof}
Since $\Talg$ is $\alpha$-competitive on tree metrics, we have that $\cost(\Scal) \leq \cost(\Scal^T) \leq \alpha\OPT(T)$.

  \begin{claim}
    \label{clm:2}
    $\OPT(T) \leq \gamma \OPT$.
  \end{claim}

  \begin{proof}
    Consider the solution on $T$ formed by embedding $\Scal^*$ into $T$ using the extension function $F$ of $f$, i.e. its responses are $(\Rhat_i, \Chat_i)$ where $\Rhat_i = F(R^*_i)$  and $\Chat_i = F(C^*_i)$. The embedded solution is feasible for the instance on $T$ and has cost at most $\sum_{u,v \in V} d_T(u,v) \rho(I^*_{u,v}) \leq \alpha \sum_{u,v \in V} d_V(u,v) \rho(I^*_{u,v}) = \gamma\OPT$ because $f$ has expansion at most $\gamma$, and the extension $F$ is expansion preserving.
  \end{proof}
  With these claims in hand, we conclude that $\alg$ is $\alpha\gamma$-competitive.
\end{proof}

\subsection{Using Online Embeddings into Other Families of Metrics}
\label{sec:using-general}

The main property of tree metrics that the proof of Theorem~\ref{thm:reduction-trees} uses is Proposition~\ref{prop:tree-containment} which enables the algorithm $\ALG$ to translate the solution on the tree $\Scal^T$ back into the original metric. While the proposition is not true for any arbitrary family $\Mcal$ of metrics, we can still apply the same approach by requiring that the online embedding algorithm $\Membed$ and the online algorithm $\Malg$ for $P$ on metrics in $\Mcal$ satisfy the following additional properties.

Consider an instance of $P$ on a metric $(V, d_V)$ with request sequence $Z_1, \ldots, Z_r$. Let $f$ be an online embedding of $Z_1, \ldots, Z_r$ into a metric $(M, d_M)$ that is fully extendable with respect to $M_1 \subseteq \ldots \subseteq M_r \subseteq M$, and $\Scal^M = ((R^M_1, C^M_1), \ldots, (R^M_r, C^M_r))$ be a solution to the instance on $(M, d_M)$ induced by $f$. Then, $f$ and $\Scal^M$ are \emph{compatible} if $M_i$ contains the vertex set of the subgraph $R^M_i$ for each $1 \leq i \leq r$. When either $f$ or $\Scal^M$ is probabilistic, then we require that the above property holds for every embedding in the support of $f$ and every solution in the support of $\Scal^M$.

Suppose $\Mcal$ is a family of metrics. Let $\Membed$ be an online embedding algorithm that embeds into $\Mcal$, and $\Malg$ be an online algorithm for $P$ on instances over metrics belonging to $\Mcal$, then $\Membed$ and $\Malg$ are compatible if for every instance of $P$ on an arbitrary metric $(V, d_V)$, algorithm $A$ produces an online embedding $f$ of the request sequence into a metric $(M, d_M) \in \Mcal$ such that the solution $\Scal^M$ of $\Malg$ on the instance induced by $f$ on $(M, d_M)$ is compatible with $f$.

\begin{theorem}\label{thm:reduction}
  Let $P$ be an online Abstract Network Design problem and $\Mcal$ be a family of metrics. Suppose that there exists an online algorithm $\Malg$ for $P$ over metrics belonging to $\Mcal$ with competitive ratio $\beta$, and that there exists an online embedding algorithm $\Membed$ that embeds into $\Mcal$ with distortion $\alpha$. Then, there exists an online algorithm $\ALG$ for $P$ over arbitrary metrics with competitive ratio $\alpha\beta$. Moreover, when either $\Membed$ or $\Malg$ is randomized, then $\ALG$ is randomized with expected competitive ratio $\alpha\beta$.
\end{theorem}

The proof of the theorem is similar to that of Theorem~\ref{thm:reduction-trees}.

\section{Online Probabilistic Embedding}\label{sec:online:emb}

In Section~\ref{sec:overview}, we give a sketch of the construction of the online embedding of Theorem~\ref{thm:dist-ub} and summarize in Theorem~\ref{thm:summary} the properties needed for Section~\ref{sec:framework}. The full details of the construction and its analysis appear in Section~\ref{online_embedding}. 

\subsection{Overview}
\label{sec:overview}
\paragraph{Constructing the online embedding.}
A \emph{$\Delta$-bounded probabilistic partition} of $V$ is a distribution over partitions $P$ of $V$, with cluster diameters bounded by $\Delta$. 
A $\Delta$-bounded probabilistic partition has \emph{padding parameter} $\gamma$ if for each $v \in V$ and any $\delta >0$ the probability that $B(v, {\delta \Delta}/{\gamma})$ is cut by the cluster of $P$ that contains $v$ is at most $\delta$.  

Let $X_k \subseteq V$ be the set of $k$ terminals revealed to the online embedding thus far. We construct a collection of nested probabilistic partitions of $X_k$,  with diameters decreasing by $\mu$. The number of such partitions is $O(\log_{\mu}\Phi(X_k))$, namely, these partitions capture all scales of distances in $X_k$ up to a factor $\mu$. Moreover, each probabilistic partition in this collection has padding parameter $O(\log k)$. The whole hierarchical structure is maintained online. Bartal showed in \cite{Bar96} that constructing such a hierarchical probabilistic partition of $X_k$ implies an embedding of $X_k$ into a distribution of HST's, with expected distortion $O(\mu \log_{\mu}\Phi(X_k) \log k)$.

The online Algorithm~\ref{online-part} in Appendix~\ref{online_embedding} maintains a  $\Delta$-bounded probabilistic partition of the current terminal set $X_k$,  for a given scale $\Delta$, with  padding parameter $O(\log k)$.  Notably, the algorithm does not assume the knowledge of $k$ upfront, as is the case in previous works. 

The construction is based on the (offline) probabilistic partitions of \cite{ABN06}. Their algorithm iteratively partitions a given metric space $V$ in the following way: At the $j$-th step, a still unclustered point $v_j$ is chosen in a particular way\footnote{Minimizing the local growth rate, see \cite{ABN06} for the details.} that is associated with a parameter $\chi_j \geq 2$; Then, the radius $r_j$ is randomly chosen from the distribution $p(r)=\frac{\chi_j^2}{1-\chi_j^{-2}}\frac{8 \ln \chi_j}{\Delta} \chi_j^{-\frac{8r}{\Delta}}$, for $r\in [\Delta/4, \Delta/2]$. The $j$-th cluster is defined to be $B(v_j, r_j)$  intersected  with the still unclustered points in $V$.  They show (Lemma $5$ in \cite{ABN06}) that this construction gives a probabilistic partition with padding parameter $O(\max_{1\leq j \leq t} \log \chi_j)$, where $t$ is the number of clusters that can be obtained by this construction. This construction cannot be implemented in an online fashion, because the value of $\chi_j$ can be changed latter, after it was already used by the algorithm to choose the radius.

 However, their analysis implicitly implies that the partition in which $v_j$ is chosen {\it arbitrarily} from all the uncovered points, and $r_j$ is chosen with $\chi_j$ being a parameter satisfying  $\sum_{ 1 \leq j \leq t} 1/\chi_j \leq 1$, has the same bound on its padding parameter (see Lemma~\ref{ABN}). This observation allows us to construct partitions in an online fashion, with $\chi_j$ being set online to satisfy the above requirement.

In particular, Algorithm~\ref{online-part} works as follows: When a new terminal $x_k$ arrives, it is added to the first cluster, by the order of construction, that contains $x_k$. If no such cluster exists, a new cluster is created by randomly picking the radius $r$, according to $p(r)$ with $\chi_k=2k^2$, and defining $x_k$ as its center and the only point. The choice of $\chi_k$ is such that the sum of $1/\chi_j$ over the $k$ steps of the algorithm is at most $1$, implying the padding parameter $O(\log k)$, as the maximal number of clusters is obviously $k$.

An online Algorithm \ref{online_hier_part} maintains the hierarchy $N$ of nested probabilistic partitions for $\Theta( \log_{\mu}\Phi(X_k) )$ different distance scales, which cover all distances in $X_k$, up to a factor $\mu$. When a new terminal $x_k$ arrives, the algorithm checks whether the aspect ratio has been increased by at least a constant factor of $\mu$, in which case it creates and adds new top (or bottom) scales to ${N}$. The algorithm then inserts the terminal $x_k$ to all the partitions, in all the levels from top to bottom, by applying the single-scale procedure of Algorithm \ref{online-part}. Eventually, the HST tree $T_k$ is constructed from the hierarchy $N$, in a standard way.

Finally, Claim~\ref{extend_HST} shows that the above online probabilistic metric-oblivious embedding is fully extendable tree embedding. 

We conclude the overview with a summary of the properties of our embedding that are used in Section~\ref{sec:framework}.

\begin{theorem}
  \label{thm:summary}
  For any metric space $(V, d_V)$, sequence of terminals $x_1, \ldots, x_k \in V$, and parameter $\mu > 1$, there exists a fully extendable
  probabilistic online embedding into a random $\mu$-HST $T$ such that
  \begin{inparaenum}[\bfseries (i)]
  \item $T$ has $O(k)$ edges;
  \item for every $u,v \in V$, and $L > 0$, we have $\Pr_T[d_T(u,v) \geq L] \leq O(\log k)\frac{d(u,v)}{L}$.
  \end{inparaenum}
\end{theorem}

The first item is proven in Observation \ref{observ:tree_size}, and the second item is proven in Claim \ref{observ:probab_ineq}.

\subsection{Construction}
\label{online_embedding}
We first give some basic definitions, and then proceed to the online constructions.
\subsubsection{Probabilistic Partition and Hierarchical Probabilistic Partition} 
For a metric space $V$, for any $\Delta >0$, a $\Delta$-bounded probabilistic partition is a distribution $\mathcal{P}$ over partitions $P$ of $V$,  such that $P= \dot{\cup}{C_j}$, $C_j \subset V$, and $\diam(C_j) \leq \Delta$.  For a partition $P$, let $P(x)$ denote the cluster that contains $x$.
Probabilistic partition $\mathcal{P}$ has padding parameter $\gamma >0$ if for all $ x\in V$ and for all $ 0< \delta < 1$, 
\[\Pr_{P \sim\mathcal{P}} \left[ B \left(x, {\delta}\Delta /{\gamma}\right) \not\subseteq P(x)  \right] \leq \delta.\]

Probabilistic partitions with a padding parameter have been studied in many works \cite{LS91,KPR93,Bar96,FRT04}.  For our construction we use the (offline) probabilistic partition of \cite{ABN06}, which is based on \cite{Bar96}. They partition a given $V$ into $\Delta$-bounded clusters iteratively: at the step $j\geq 1$, the point $v_j$ is chosen from the still unclustered points, with the minimum local growth rate $\chi_j \geq 2$ (see \cite{ABN06} for the precise definition of the local growth rate). Then, a random radius $r_j$ is chosen from the the distribution given by $p(r)=\frac{\chi_j^2}{1-\chi_j^{-2}}\frac{8 \ln \chi_j}{\Delta} \chi_j^{-\frac{8r}{\Delta}}$, where $r\in [\Delta/4, \Delta/2]$. The cluster $C_j$ is defined to be the intersection of $B(v_j, r_j)$ with the points that are still not clustered. The padding parameter is shown to be $\max_{1\leq j\leq t} \{\log \chi_j\}$, where $t$ is the number of clusters in partition. 

Inspecting the analysis of the above partition more carefully, it can be observed that choosing  $v_j$ {\it arbitrarily} among the unclustered points, and choosing the radius $r_j$ with $\chi_j$ that satisfy $\sum_{ 1\leq j \leq t}\chi_j^{-1} \leq 1$, implies the random partition with the same padding property. 

Let $\mathcal{\hat{P}}$ denote the random partition constructed in that way. For a partition $P \in \mathcal{\hat{P}}$, let $\{C_j\} \subset P$ denote its clusters. The following technical lemma is implicitly proved in \cite{ABN06}(Lemma $5$):  

\begin{lemma}\label{ABN}
For $x \in V$, let $j$ be such that $x \in C_j$. For any $1/2 < \theta \leq 1:$ 
\[\Pr_{P \sim \mathcal{\hat{P}}}\left[ B\left(x, \frac{\ln (1/\theta)}{32 \ln \chi_j}\cdot \Delta\right) \not\subseteq C_j \right] \leq (1- \theta)\left ( 1+ \theta \sum_{l=1}^{t} 1/\chi_l\right),\]
where $\chi_j \geq 2$ is the parameter used to pick the radius $r_j$ of the cluster $C_j$, and $t$ is the number of clusters in $P$.
\end{lemma}

We will use this observation to simulate the random partition of \cite{ABN06} in an online way.\\

\noindent
{\bf Hierarchy of probabilistic partitions.}
We define the notion of a \emph{hierarchical} probabilistic partition, which is basically a collection of probabilistic partitions that cover all the distance scales of the current pointset.  

For a given finite metric space $V$, let $A \leq B$ be two integers. Consider the set of distance scales $\{\Delta_j|  A\leq j \leq B\}$, such that $\Delta_A=d_{\max}(V)$, $\Delta_B=d_{\min}(V)$ and for $A \leq  j <  B$, $\Delta_{j+1} =\Delta_j/{\mu}$.  
A  {\it hierarchical probabilistic partition} $\mathcal{H}$ of $V$ with parameter $\mu>1$ is a collection of nested $\Delta_j$-bounded probabilistic partitions $\mathcal{P}_j$ of $V$, for $ A\leq j \leq B$.  The partitions are nested, i.e, each cluster $C$ of each partition $P_j \sim \mathcal{P}_j$ is (randomly) partitioned by clusters ${P}_{j+1}[C]$ (the probabilistic partition $\mathcal{P}_{j+1}$ induced on the cluster $C$) of each partition $P_{j+1} \sim \mathcal{P}_{j+1}$.  If for each cluster $C \in P_j$, for each $P_j \sim \mathcal{P}_j$ the probabilistic partition $\mathcal{P}_{j+1}[C]$   has padding parameter $\gamma_j$ then the padding parameter of the whole hierarchy is defined by $\gamma(\mathcal{H}) = \sum_{A \leq j  \leq B} \gamma_j$.

In \cite{Bar96}, among other results,  Bartal showed that:

\begin{theorem}[Theorem 13, \cite{Bar96}]\label{hst_emb_bar} Given a hierarchical probabilistic partition $\mathcal{H}$ of $V$ with parameter $\mu >1$, one can construct a randomized embedding into an $\mu$-HST tree, with expected distortion $O(\mu \cdot \gamma(\mathcal{H}))$. \end{theorem}

The $\mu$-HST tree is naturally derived from the nested structure of the hierarchy, we describe the construction in the proof of Theorem \ref{thm:dist-ub}.

In Appendix \ref{online_hierarchy} we describe an online randomized algorithm that maintains a hierarchical partition $H$ of the current point set $X_i$, with padding parameter $O(\mu \log_{\mu}\Phi(X_i)\log |X_i|)$. For each relevant scale of distances in the current point set, there is a bounded partition for this scale, which is maintained online as well. The algorithm does not assume a prior knowledge on the underlying metric space, it rather uses only information on the current point set.

\subsubsection{Online $\Delta$-Bounded Probabilistic Partition }
We present the online construction of a $\Delta$-bounded partition of a subspace $Z\subseteq V$, the points of which are revealed to the algorithm one by one. For a given $\Delta$, Algorithm \ref{online-part} maintains a random $\Delta$-bounded partition ${P} \in \mathcal{P}$ (sampled from a distribution of $\Delta$-bounded partitions $\mathcal{P}$) for the current subspace $Z$. The partition $P$ is a collection of clusters of the form $C(v,r)$, each cluster is represented by the pair $(v, r)$, where  $v \in Z$ is the center of the cluster and $r>0$ is its radius. Particularly, $C(v,r)$ is the subset of all points in $Z$ within distance at most $r$ from $v$.  When a new point $z \in V \setminus Z$ is given to the algorithm, ${P}$ is (randomly) updated to be $\Delta$-bounded partition of $Z \cup \{z\}$. At each time step the padding parameter of the partition is $O(\log |Z|)$. 
 At the beginning $X_0=\emptyset$, ${P} = \emptyset$ and $t=0$ ($ 0\leq t \leq |Z| $ counts the current number of clusters in $P$). The argument $P$ identifies a particular partition of a subspace $Z$ that the algorithm maintains.  
\begin{algorithm}\label{one:online:partition}
\caption{$\Delta$-Bounded Online Probabilistic Partition,  $\Delta$-${\bf BOPP}$$\left< P \right>$}
\label{online-part}
Let $P=\dot{\cup}_{j=1}^{t}C(v_j, r_j)$ denote the current set of clusters in $P$.
When a new point $z \in V \setminus Z$ arrives:
\begin{algorithmic}[1]
\FORALL{ $1\leq j \leq t$} 
\IF{$d(v_j, z)\leq r_j$} \STATE{update $C (v_j, r_j)\leftarrow C(v_j, r_j) \cup \{z\}$;} \STATE terminate\ENDIF 
\ENDFOR
\STATE Set $t \leftarrow t+1$ and $\chi_{t}=2 t^2$. Pick $r$ distributed according to $p(r)=\frac{(\chi_{t})^2}{1-(\chi_{t})^{-2}}\frac{8 \ln (\chi_{t})}{\Delta} (\chi_{t})^{-\frac{8r}{\Delta}}$, for $r \in [\Delta/4, \Delta/2]$, independently from the previous iterations. Set $C(v_{t}, r) \leftarrow \{z\}$. Update $P \leftarrow P\cup C(v_{t}, r)$.
\end{algorithmic}
\end{algorithm}

\begin{lemma}\label{prob_part}
Let $\mathcal{P}$ be a random partition of $Z$, constructed by Algorithm \ref{online-part}, starting from $Z= \emptyset$ and revealing the points of $Z$ one by one. The padding parameter of $\mathcal{P}$ is $O(\log |Z|)$.
\end{lemma}

\begin{proof}
The main observation is that the random partition $\mathcal{P}$ of $Z$ constructed by the above online algorithm has the same distribution as the random partition $\mathcal{\hat{P}}$ of $Z$ constructed offline, where the points of $Z$ are given upfront. Therefore, the bound in Lemma \ref{ABN} holds true for the partition $\mathcal{P}$ as well.  

For any $x \in Z$ and any $0<\delta <1$, we show that $\Pr_{P \sim \mathcal{P} }[B(x, \delta \Delta /(c \cdot \log |Z|)) \not\subseteq P(x)] \leq \delta$, for some constant $c$. 
By Lemma \ref{ABN}, for any $1/2 \leq \theta <1$, we have:
\[\Pr_{P \sim \mathcal{P}}\left[ B\left(x, \frac{\ln (1/\theta)}{32 \ln \chi_j}\cdot \Delta\right) \not\subseteq P(x) \right] \leq (1- \theta)\left ( 1+ \theta \sum_{l=1}^{t} 1/\chi_l\right),\]
where $j$ is such that $P(x)=C(v_j, r_j)$, and $t$ is the number of clusters in $P$. Since $t \leq |Z|$ and since for all $ 1\leq l \leq t$, $\chi_l = 2 l^2$, it holds that:
\[\sum_{l=1}^{t} 1/\chi_l = \sum_{l=1}^{t} 1/(2l^2) < \frac{\pi^2}{12}<1.\]

Therefore, \[\Pr_{P \sim \mathcal{P}}\left[ B\left(x, \frac{\ln (1/\theta)}{32 \ln \chi_j}\cdot \Delta\right) \not\subseteq P(x) \right] \leq 1-\theta^2.\]
Let $\delta=\frac{4}{3}(1-\theta^2)$, we have  
\[\Pr_{P \sim \mathcal{P}}\left[ B\left(x, \frac{\ln (1/(1-\frac{3}{4}\delta))}{64 \ln \chi_j}\cdot \Delta\right) \not\subseteq P(x) \right] \leq \frac{3}{4}\delta.\]

Note that $1/2 \leq \theta <1$ and therefore $0<\delta\leq 1$.
Since for any $z>0$, $\ln (1/(1-z)) \geq z$, it holds that $B\left(x, \frac{\frac{3}{4}\delta}{64 \ln \chi_j}\cdot \Delta\right) \subseteq B\left(x, \frac{\ln (1/(1-\frac{3}{4}\delta)}{64 \ln \chi_j}\cdot \Delta \right)$, implying 
\[\Pr_{P \sim \mathcal{P}}\left[B\left(x, \frac{\delta}{\frac{256}{3} \ln \chi_j}\cdot \Delta\right) \not\subseteq P(x)\right]
\leq \Pr_{P \sim \mathcal{P}}\left[ B\left(x, \frac{\ln (1/(1-\frac{3}{4}\delta))}{64 \ln \chi_j}\cdot \Delta\right) \not\subseteq P(x) \right] \leq \frac{3}{4}\delta.\]

Finally, since for all clusters $C_j \in P$, $\chi_j \leq 2|Z|^2$, it holds that for some constant $c>0$
\[\Pr_{P \sim \mathcal{P}}\left[B\left(x, \frac{\delta}{c \cdot \ln (2|Z|^2)}\cdot \Delta\right)\not\subseteq P(x)\right]
\leq \Pr_{P \sim \mathcal{P}}\left[B\left(x, \frac{\delta}{\frac{256}{3} \ln \chi_j}\cdot \Delta\right)\not\subseteq P(x)\right] \leq \delta,\]
which completes the proof.
\end{proof}

\subsubsection{Online Hierarchical Probabilistic  Partition}\label{online_hierarchy}
Algorithm \ref{online_hier_part} maintains the hierarchical partition $N$ of the current terminal set $X_i$, with $\mu>4$.  The hierarchy $N$ is randomly sampled from a distribution $\mathcal{H}$. The scale set of this hierarchical partition is $\{\Delta_j| A\leq j\leq B\}$, where $A < B$ are some integers. At each time step $i$, the scales satisfy the following properties:  $4d_{\max}(X_{i}) \leq \Delta_{A} \leq \mu 4 d_{\max}(X_{i})  $, $4 d_{\min}(X_i)\leq \Delta_B \leq \frac{4 d_{\min}(X_i)}{\mu}$, and for all $A\leq  j < B$, $\Delta_{j+1}=\Delta_{j}/\mu$. The number of scales is $\lceil  \log_{\mu}\Phi(X_i) \rceil +1$. 

Essentially, the algorithm maintains randomly constructed $\Delta_j$-bounded partitions $P_j$, for each scale $\Delta_j$, such that the clusters of the partition $P_{j+1}$ are $\Delta_{j+1}$-bounded random partitions of the clusters of $P_j$. Moreover,  for each cluster $C$ of the partition $P_j$ the random partition $\mathcal{P}_{j+1}[C]$ of $C$ has padding parameter $O(\log |C|)=O(\log i)$.
The algorithm has to maintain the current set of the relevant distance scales: When a new terminal $x_i$ arrives the algorithm checks whether the aspect ratio has increased by at least a constant factor of $\mu$ with respect to the current upper or lower scale level. If so, it adds new, either top or bottom, scales to ${N}$.  While creating new scales, the algorithm keeps the structure of $N$ to be nested. The largest scale is defined such that its only cluster contains all the points of the current set $X_i$.  Next, algorithm adds $x_i$  to the relevant partitions of all scales in ${N}$, using Algorithm \ref{online-part}, while keeping the partitions to be nested as well.  Particularly, for each cluster $C$ of a partition of scale $j$ there is a partition of it in level $j+1$, denoted by $P_{j+1}[C]$, which is maintained online by algorithm $\Delta_{j+1}$-{\bf BOPP}$\left <P_{j+1}[C] \right>$. 

At the beginning $X_2=\{x_1, x_2\}$, $A=0, B=1$, $\Delta_A=4d(x_1, x_2)$, $\Delta_B=\Delta_A/\mu$. Set $P_A, P_B \leftarrow \emptyset$. Then, update these partitions as follows: apply $P_A \leftarrow \Delta_{A}$-{\bf BOPP}$\left <P_A\right>$ on $x_1$ to add it to the partition $P_A$ , and then apply $\Delta_{A}$-{\bf BOPP}$\left <P_A\right>$ on $x_2$ to add it to $P_A$. After these steps $P_A$ contains one cluster, that contains both $x_1$ and $x_2$. Similarly, apply $P_B \leftarrow \Delta_{B}$-{\bf BOPP}$\left <P_B[P_A(x_1)]\right>$ (recall that $P_A(x_1)$ denotes the cluster in partition $P_A$ that contains $x_1$) on $x_1$ to add it the partition of the cluster $P_A(x_1)$ at level $B$, and then apply $\Delta_{B}$-{\bf BOPP}$\left <P_B[P_A(x_2)]\right>$.

\begin{algorithm}[h]
    \caption{Online Hierarchical Probabilistic Partitions}
    \label{online_hier_part}

\begin{algorithmic}[1]
\STATE When a new terminal $x_i$ arrives:
 
    \IF{$d_{\max}(X_i) \geq \frac{\mu}{4} \Delta_A$}  
      \STATE Let $l = \lceil \log_{\mu} (4{d_{\max}(X_i)}/\Delta_A) \rceil$, and update $A \leftarrow A-l$. Define $P_{A-1} =\{X_i\}$.    \FORALL {$ A\leq j  < A+l$} 
			 \STATE create new $\Delta_{j}= \Delta_{A+l} \cdot \mu^{(A+l-j)}$- bounded partition ${P}_j =\emptyset$;
			\FORALL{$ 1\leq s \leq i-1$}
			\STATE apply $\Delta_j$-{\bf BOPP}$\left< P_j[P_{j-1}(x_1)]\right>$ on input $x_s$. \ENDFOR
			\ENDFOR

	\ENDIF
		\IF {${d_{\min}({X_i})} \leq \frac{\mu}{4}\Delta_B$} 
		\STATE Let $l= \lceil\log_{\mu} ( \Delta_{B}/(4{d_{\min}({X_i})}))\rceil + 1$, and update $B \leftarrow B+l$. \FORALL {$ B -l <  j \leq B$} \STATE create new $\Delta_{j}=\Delta_{B-l}/{\mu}^{j-B+l}$ bounded partition ${P}_j=\emptyset$.
		\FORALL{$ 1\leq s \leq i-1$} 
			\STATE apply $\Delta_j$-{\bf BOPP}$\left< P_j(P_{j-1}(x_s))]\right>$ on input $x_s$. \ENDFOR
		\ENDFOR
		\ENDIF \\
	  
		\STATE Create $P_{A-1} =\{X_i\}$. \STATE {\bf For all $ A \leq j \leq  B$:} apply  $\Delta_j$-{\bf BOPP}$\left< P_j[P_{j-1}(x_i)]\right>$  on input $x_i$.
	
\end{algorithmic}	
\end{algorithm}

\begin{observation}\label{observ_online}
We observe that Algorithm \ref{online_hier_part} is a proper online algorithm for maintaining hierarchical partition, i.e., updating the hierarchical partition is made in a way that does not split the previously constructed clusters: at each time step $i$ each cluster that has been already created either stays unchanged or receives the new point.
\end{observation}

\begin{lemma}\label{hier_part}
Let $\mathcal{H}$ be a random collection of partitions that was constructed by Algorithm \ref{online_hier_part}, when applied on a terminal set $X_i$ of size $i$. Then, $\mathcal{H}$ is a hierarchical probabilistic partition and its padding parameter is bounded by $O(\log_{\mu}\Phi(X_i) \log i)$.
\end{lemma}
\begin{proof}
First, note that for every $N \sim \mathcal{H}$ the partitions of all scales are nested by construction. 
Every $N \sim \mathcal{H}$ has a scale set such that $\Delta_A \leq 4\mu d_{\max}(X_i) $ and $\Delta_{B} \leq 4d_{\min}/\mu$, with the scales decreasing exactly by a factor of $\mu$. The number of the scales after $i$ steps is $O(\log_{\mu} \Phi(X_i))$. 

Since each random partition in $\mathcal{H}$ has been constructed using the online algorithm Algorithm \ref{online-part}, it holds that each probabilistic partition $\mathcal{P}_{j+1}[C]$, of each cluster $C \in P_j \sim \mathcal{P}_{j}$ has padding parameter $O(\log |C|)=O(\log i)$. This completes the proof.
\end{proof}

\subsubsection{Online Probabilistic Tree Embedding}
\label{app:constr}
\begin{proofof}{Theorem~\ref{thm:dist-ub}} In what follows we describe the construction of the metric-oblivious tree embedding. The extension of the embedding to $V$ is presented in Claim~\ref{extend_HST}. 

In \cite{Bar96} Bartal gave a natural algorithm to construct an HST tree $T$ from a given hierarchical partition $N$ of a metric space $X$. Let $P_A, \ldots, P_B$ denote the nested partitions of $N$. Recall that we assumed that $P_A$ contains all the points of $X$ as its only cluster. Let $T_1, \ldots, T_s$ be the HST trees, recursively constructed for $s$ clusters of the partition $P_{A}$. Let $r$ be the root of $T$, and define the trees $T_1, \ldots, T_s$ to be its direct children. The length of each edge connecting the root $r$ to each of its children is defined to be $\Delta_A/2$.  Note that the points of $X$ are at the leaves of $T$, and that $T$ is indeed a $\mu$-HST tree. 

Thus, using the above construction on the hierarchical partition $N$ randomly generated by Algorithm~\ref{online_hier_part}, when applied on the current terminal set $X_i$, we obtain the appropriate $\mu$-HST tree $T_i$. It remains to ensure that we can keep the tree updating process in online manner, i.e., the tree of a time step $i$ is a subtree of the tree constructed in the time step $i+1$, which is true by Observation \ref{observ_online}.
To conclude the proof  of the theorem, we apply Lemma~\ref{hier_part} to Theorem~\ref{hst_emb_bar}.
\end{proofof}
\begin{remark}\label{remark:tree_size}
We note that the resulting tree $T_i$, constructed from a hierarchical partition $N$ of $X_i$, can contain ``redundant'' paths:  these are the chain-paths from a node $v$ to $u$, created from partitions that were added to $N$ when the aspect ratio of the current set has been increased. We can compress these chains to one edge without changing the metric defined by the tree. In addition, the compressed tree will be a proper $\mu$-HST tree as well, with the number of nodes bounded by $O(|X_i|)$. 
\end{remark}

Recall the definition of a fully extendable online tree embedding, given in Definition \ref{def_fully_tree_ext}.
We prove the following:
\begin{claim}\label{extend_HST} 
Let $V$ be a metric space and $X=\{x_1, \ldots, x_k\} \subseteq V$ be a sequence of terminals. The online probabilistic metric-oblivious embedding of Theorem~\ref{thm:dist-ub} is a fully extendable online tree embedding. 
\end{claim}

\begin{proof}
For each $1\leq i \leq k$, let $f_i\colon X_i \to T_i$ denote the embedding of the step $i$, where $T_i$ is the randomly constructed $\mu$-HST tree by the algorithm of Theorem~\ref{thm:dist-ub}. Let $N$ denote the corresponding random hierarchical partition of $X_i$, from which the tree $T_i$ is obtained.  

The extension $F_i: V \to T_i$ is constructed as follows. 
For all $v\in V \setminus X_i$, go over the hierarchy $N$ to find the clusters which would contain $v$, if $v$ was the next point to arrive. In a more detail, going recursively from the topmost scale of $N$ to the lowest, for each cluster $C(u,r)$ in partition $P_j$, according to the order of construction, check whether $v$ belongs to it, i.e. $d(u,v) \leq r$. Let $ A\leq j \leq B$ be the first level such that $v$ does not belong to any cluster of the partition $P_j$. If $j=A$, then define $F_i(v)$ be the root of $T_i$. For $j>A$,  
consider the cluster $C=P_{j-1}(v)$ of partition $P_{j-1}$ that contains $v$. By construction of $T_i$ from the hierarchy $N$, there is an internal node $z_C$ on $T_i$ that corresponds to the cluster $C$. Define $F_i(v)=z_C$.  

To analyze the expansion of $F_i$, note that the construction of $F_i$ exactly simulates the online random partition of $X_i \cup v$, as if $v$ was given upfront, except that the construction does not open up new clusters, but instead matches the point $v$ to the higher cluster it belongs to in the hierarchy. This implies that at each partition level $j$ that has a cluster containing $v$, the padding property holds for the point $v$, with the same padding parameter as for the terminal points, and the distance on the tree $T_i$ between any two images $F_i(v) \neq F_i(w)$ is up to a constant factor as the distance on the tree that would have been constructed if $v$ and $w$ were given to the online embedding upfront.  

More formally, for any scale level $j\geq A$ such that the partition $P_j$ of $N$ has a cluster $C(u, r)$ with $d(u, v)\leq r$, it holds that for any point $w \in V$, the probability that the ball $B_V(v, d(v, w))$ cut by the cluster $C(u,r)$ is bounded by $O(\log i)d(v,w)/\Delta_j$, where $\Delta_j$ is the diameter of $C(u,r)$. Therefore, the expected expansion of a pair of points that were not  mapped to the same node is bounded by $O(\log i \log \Phi(X_i))$. The same holds true for pairs $(v, x_s)$, where $v\in V\setminus X_i$ and $x_s \in X_i$. For the points that were mapped to the same node on the tree by $F_i$, the expansion is $0$.

The extension $H_i: T_i \to V$ of $f_i^{-1}$ is defined as follows. Going from the leaves of $T_i$ up to the root, for each node $v \in T_i$: if $H_{i-1}(v)$ is already defined then let $H_i(v)=H_{i-1}(v)$, otherwise, if $v$ is a leaf then let $H_i(v)=f_i^{-1}(v)$ and for each internal non-leaf node $v$ let $H_i(v)=H_i(u)$ where $u$ is any child of $v$. Note that by the construction $H_{i+1}$ extends $H_{i}$, since $T_{i}$ is a subtree of $T_{i+1}$. Also, since  $f_i$ is non-contractive on the terminal set $X_i$, for any  two terminals $x_j \neq x_l$ we have $ d(x_j, x_l)\leq d_{T_i}(f_i(x_j), f_i(x_l))$, i.e., $d(H_i(f_i(x_j)), H_i(f_i(x_l))) \leq d_{T_i}(f_i(x_j), f_i(x_l))$. Next we show that $H_i$ is non-expansive on any two internal nodes $u \neq v \in T_i$. Let $f_i(x_j)$ be any leaf in the subtree rooted on $v$, and $f_i(x_l)$ be any leaf in the subtree rooted on $u$. Since $T_i$ is a $\mu$-HST it holds that $ (1/4) d_{T_{i}}(f_i(x_j), f_i(x_l)) \leq d_{T_i}(u, v)$, implying that $H_i$ is non-expansive up to a factor of $4$. Scaling up the edge costs of $T_i$ by $4$ gives  a non-expansive embedding. 
\end{proof}

By the above construction and Remark~\ref{remark:tree_size} we have:
\begin{observation}\label{observ:tree_size}
The extension $F_k$ of the online embedding $f_k: X_k \to T_k$ maps the points of $V$ into an $HST$ tree with $O(k)$ nodes.
\end{observation}

We prove the following claim, which is used in the analysis of our main theorem:
\begin{claim}\label{observ:probab_ineq}
The extension $F_k$ of the online embedding $f_k: X_k \to T_k$ is such that for any $v \neq u \in V$, and $L>0$, we have $Pr_{T_k}[d_{T_k}(u,v) \geq L] \leq O(\log k)\frac{d(u,v)}{L}$.
\end{claim}

\begin{proof}
As we have shown in the proof of Claim \ref{extend_HST}, at each level $A\leq j \leq B$ such that there is a cluster $C \in P_j$ that contains $v \in V$, it holds that $Pr[B(v, d(v, w)) \not \subseteq C] \leq O(\log k) d(v,w)/\Delta_j$, for any $w \in V$. Note that $B(v, d(v, w)) \not \subseteq C$ implies that $v$ and $w$ will be separated at the level $j$, and will be mapped to two nodes on the tree $T_k$ with the distance $d_{T_k}(u,v)=O(\Delta_j)$. Therefore, by the union bound, and since $T_k$ is an HST tree, we have $Pr_{T_k}[d_{T_k}(u,v) \geq L] \leq O(\log k)\frac{d(u,v)}{L}$.
\end{proof}

\section{Lower Bounds on Online Embedding into Trees}\label{sec:lowerbound}
In this section we prove lower bounds on the expected distortion of probabilistic online embeddings into trees (Theorems~\ref{thm:dist-lb} and \ref{thm:dist-HST-lb}). Using the standard approach for proving lower bounds against probabilistic constructions (Yao's Principle), it suffices to construct an underlying metric space $(V_\ell, d_{V_\ell})$ together with a distribution $\Dcal_\ell$ of terminal sequences $\sigma$ and vertex pairs $(u,v)$ such that for any deterministic online embedding into any HST (general tree) $T$, we have that
$E_{(\sigma, (u,v)) \in \Dcal_\ell}[d_T(u,v)/d_{V_\ell}(u,v)]$ 
satisfies the desired lower bound.

Both lower bounds follow the same general framework which we present below.

\subsection{Lower Bound Framework}
We will need the following result of Rabinovich and Raz~\cite{RabinovichR98}:
\begin{theorem}[\cite{RabinovichR98}]
  \label{thm:RR}
  Any embedding of the $n$-vertex cycle into a tree has distortion at least $n/3-1$.
\end{theorem}

\begin{definition}[Base graph]
  A \emph{base graph} $B$ of \emph{width}
 $\phi$ is an unweighted connected graph with distinguished vertices $s,t$ and can be decomposed into edge-disjoint $(s,t)$-paths such that each path is an $(s,t)$-shortest path in $B$ and has exactly $\phi$ edges. The vertex $s$ is called the source and $t$ the sink.
\end{definition}

\paragraph{Constructing the underlying metric space.}The metric space $(V_\ell,d_{V_\ell})$ will be the shortest-path metric of a weighted graph $\Ghat_\ell = (V_\ell,E_\ell)$ which is constructed recursively using some suitable base graph $B$. Each graph $\Ghat_\ell$ will consist of a source vertex $s_\ell$ and a sink vertex $t_\ell$. For the base case, the graph $\Ghat_0$ consists of a single edge $(s_0, t_0)$. The graph $\Ghat_\ell$ is constructed by replacing each edge of ${\Ghat}_{\ell-1}$ with a scaled-down copy of $B$ in which each edge has length $\phi^{-\ell}$. More precisely, for each edge $(u,v)$ of $\Ghat_{\ell-1}$, we remove $(u,v)$, add a scaled-down copy of $B$ in which each edge has length $\phi^{-\ell}$, and contract $u$ with the copy's source and $v$ with the copy's sink.

\paragraph{Constructing $\Dcal_\ell$.}
The distribution $\Dcal_\ell$ is over pairs $(\sigma, e)$ where $\sigma$ is a sequence of vertices in $V_\ell$ and $e$ is an edge of $\Ghat_\ell$. Consider the following process for constructing a sequence of random edge-weighted graphs $G_0, \ldots, G_\ell$. For each $G_i$, a random subset of the edges are called \emph{active}. Initially, $G_0$ consists of a single active edge. For $i > 0$, $G_i$ is constructed by taking $G_{i-1}$ and replacing a uniformly random active edge $(u,v)$ with a scaled-down copy of $B$ in which each edge has length $\phi^{-\ell}$; the newly-introduced edges become active and all others become inactive. The sequence $\sigma$ is the vertices of $G_\ell$ in the order in which they were introduced: the vertices of $G_0$ is presented first, then the new vertices of $G_1$, etc. The random edge $e$ is a uniformly random edge among the active edges of $G_\ell$, i.e.~among the ones introduced last.\\

\begin{figure}
  \centering
  \begin{subfigure}[b]{0.15\textwidth}
    \centering
    \includegraphics[scale=0.4]{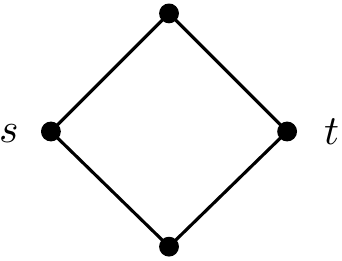}
    \caption{$B$}
  \end{subfigure}
\hfill
  \begin{subfigure}[b]{0.15\textwidth}
    \centering
    \includegraphics[scale=0.4]{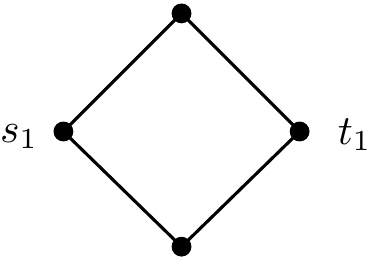}
    \caption{$\Ghat_1$}
  \end{subfigure}
\hfill
  \begin{subfigure}[b]{0.15\textwidth}
    \centering
    \includegraphics[scale=0.4]{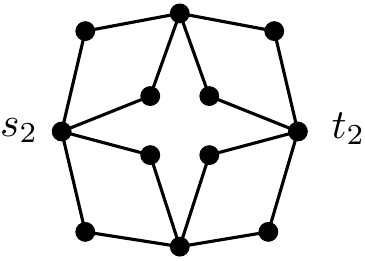}
    \caption{$\Ghat_2$}
  \end{subfigure}
\hfill
  \begin{subfigure}[b]{0.15\textwidth}
    \centering
      \includegraphics[scale=0.4]{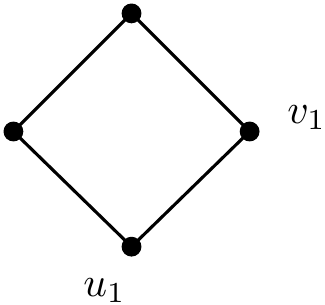}
      \caption{$G_1$}
  \end{subfigure}
\hfill
  \begin{subfigure}[b]{0.15\textwidth}
    \centering
      \includegraphics[scale=0.4]{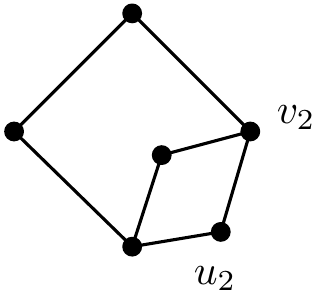}    
      \caption{$G_2$}
  \end{subfigure}
\hfill
  \begin{subfigure}[b]{0.15\textwidth}
    \centering
      \includegraphics[scale=0.4]{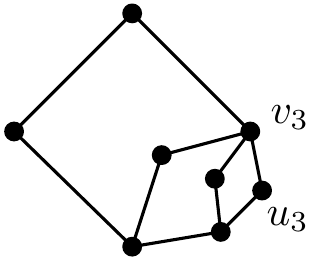}    
      \caption{$G_3$}
  \end{subfigure}
  \caption{\small Example of $\Ghat_\ell$ and $G_\ell$ with a $4$-cycle as base
    graph $B$. Here, $\phi = 2$ so each edge of $\Ghat_\ell$ has length $2^{-\ell}$.}
  \label{fig:diamond-underlying}
\end{figure}

\paragraph{Constructing the base graph.}
In the following, we use the convention that an unweighted graph $G$ induces a metric $d_G$ over its vertex set where $d_G(u,v)$ is the length of the shortest $(u,v)$-path in
$G$ assuming that each edge of $G$ has unit length. In particular, a non-contractive embedding of $G$ into a tree $T$ is one with $d_T(e) \geq 1$ for every edge $e$ of $G$.

\begin{lemma}
  \label{lem:girth-dist}
Let $G = (V,E)$ be an unweighted $n$-vertex graph with girth $g$. Consider a non-contractive embedding of $G$ into a tree $T$. If we choose an edge $e$ in $E$ uniformly at random, then $\Pr[d_T(e) \geq  (g/3)-1] \geq (|E| - (n-1))/|E|$.
\end{lemma}

\begin{proof}
Let $F = \{e \in E : d_T(e) \geq ((g/3)-1) \cdot d_V(e)\}$. Since we pick an edge $e \in E$ uniformly at random, we have that $\Pr[d_T(e) \geq (g/3)-1] = |F|/|E|$. On the one hand, Theorem~\ref{thm:RR} implies that $F$ intersects all cycles of $G$---i.e. $F$ is a feedback edge set---and so $E \setminus C$ has no cycles; thus, $|E| - |F| \leq n - 1$. Therefore, $|F| / |E| \geq (|E| - (n-1))/|E|$, as desired.
\end{proof}

\begin{lemma}
  \label{lem:base-HST}
  For some constant $c > 0$, there exists an infinite family $\{B_m\}$ of graphs that satisfies the following properties:
  \begin{inparaenum}[\bfseries (1)]
  \item $B_m$ is a base graph with $m$ vertices and $O(m)$ edges;
  \item \label{base-dist} $B_m$ has width $\phi_m = \Theta(\log m)$;
  \item $B_m$ has at least $3m$ edges;
  \item $B_m$ has girth $\Theta(\log m)$.
  \end{inparaenum}
\end{lemma}

We also remark that the natural choice of diamond graphs~\cite{NewmanR02} does not satisfy the second property as it has $\phi_m = \Theta(\sqrt{m})$.

\begin{proof}
  Our approach is to construct a multi-partite graph with logarithmic girth that contains many edge-disjoint paths from the first part to the last, and then add vertices $s$ and $t$ and connect them to the first and last parts, respectively, using paths of logarithmic length.

  \begin{claim}
  For all sufficiently large $n$ and $3 \leq D \leq 7$, there exists an $L_n$-partite graph $H_n$ with $L_n = \lfloor \log n \rfloor + 1$ parts, whose number of vertices $|V(H_n)|$ satisfies $\frac{D-1}{D}nL_n \leq |V(H_n)| \leq nL_n$, and whose edge set comprises $(D-1)n$ edge-disjoint paths from the first part to the last, and has girth at least $g_n = (L_n - 1)/4 = \lfloor \log n \rfloor/4$.
\end{claim}
\begin{proof}
We construct $H_n$ using the probabilistic method. Consider the following random $L_n$-partite graph $\mathcal{H}_n$ with $n$ vertices in each part. Let $V_i$ denote the $i$-th part. For each $i$, the edges from $V_i$ to $V_{i+1}$ is an independent uniformly random $D$-regular bipartite graph $\mathcal{H}_n^i$. Observe that $\mathcal{H}_n$ contains $Dn$ edge-disjoint paths from the first part to the last. 

In the following, we say that a simple cycle is \emph{short} if its length is less than $g_n$. We now show that, with positive probability, it suffices to remove at most $n$ edges from $\mathcal{H}_n$ to remove all short cycles, and we do this by bounding the expected number of short cycles. Consider the probability that a given cycle $C$ of length $2\ell \leq g_n$ appears in $\mathcal{H}_n$ (note that any cycle in $\mathcal{H}_n$ will be of even length). Denote by $e_1, \ldots, e_{2\ell}$ and $v_1, \ldots, v_{2\ell}$ be the edges and vertices, respectively, of $C$, where $e_j = (v_j, v_{j+1})$, where we use $v_{2\ell + 1}$ to also mean $v_1$. Let $E_j = \{e_1, \ldots, e_j\}$. We have that $\Pr[C \subseteq \mathcal{H}_n] = \Pr[e_j \in \mathcal{H}_n] \cdot \prod_{j=2}^{2\ell} \Pr[e_j \in \mathcal{H}_n \mid E_{j-1} \subseteq \mathcal{H}_n ]$. We now bound $\Pr[e_j \in \mathcal{H}_n \mid E_{j-1} \subseteq \mathcal{H}_n ]$ for some $j$. Suppose $e_j \in V_i \times V_{i+1}$, and suppose without loss of generality that $v_j \in V_{i}$ and $v_{j+1} \in V_{i+1}$. After conditioning on $E_{j-1} \subseteq \mathcal{H}_n$, the edges in $\mathcal{H}_n^i$ that are incident to $v_j$ are uniformly random among the remaining options. There are two cases to consider: (1) either $v_i$ is not incident to an edge in $E_{j-1} \cap \mathcal{H}_n^i$; (2) or it is. In the first case, the remaining $D$ neighbors of $v_i$ in $\mathcal{H}_n^i$ are uniformly random among all of $V_j$. This is because $D > 2$ and each vertex of $V_j$ is incident to at most two edges of $E_{j-1}$. Similarly, in the second case, the remaining $D-1$ neighbors of $v_i$ are uniformly random among $n-1$ vertices of $V_j$. Thus, we have 
\[\Pr[e_j \in \mathcal{H}_n \mid E_{j-1} \subseteq \mathcal{H}_n ] \leq \max\left\{\frac{D-1}{n-1}, \frac{D}{n}\right\} = \frac{D}{n}.\] and so
\[\Pr[C \subseteq \mathcal{H}_n] \leq \left(\frac{D}{n}\right)^{2\ell}.\]

Next we bound the total number of possible simple cycles of length $2\ell$. There are $nL_n$ choices for $v_1$. Once we have chosen $v_{j-1}$, there are at most $2n$ choices for $v_j$: if $v_{j-1} \in V_i$, then $v_j$ must be in either $V_{i-1}$ or $V_{i+1}$. Thus, the number of cycles of length $2\ell$ is at most $nL_n \cdot (2n)^{2\ell-1} \leq L_n \cdot (2n)^{2\ell}$. Together with the above, we get that the expected number of short cycles is at most
\begin{align*}
  \sum_{\ell = 1}^{g_n/2} L_n \cdot (2n)^{2\ell} \cdot  \left(\frac{D}{n-1}\right)^{2\ell}
  &=  L_n \sum_{\ell = 1}^{g_n/2}(2D)^{2\ell}\\
  &= L_n \cdot O((2D)^{g_n})\\
  &= L_n \cdot O(14^{(\log n)/4})\\
  &= o(n).
\end{align*}

Therefore, there exists a graph $\hat{H}_n$ in the support of $\mathcal{H}_n$ with at most $n$ short cycles. Since $\hat{H}_n$ has at most $n$ short cycles, we can remove a set $F$ of at most $n$ edges to make $\hat{H}_n$ have girth at least $g_n$. Let $\hat{\mathcal{P}}_n$ be the collection of $Dn$ edge-disjoint paths of length $L_n -1$ from the first part to the last and let $\mathcal{P}_n$ be the sub-collection of paths that do not contain any edge from $F$. We have that $|\mathcal{P}_n| \geq |\hat{\mathcal{P}}_n| - |F| \geq (D-1)n$. We obtain $H_n$ by taking the union of  arbitrary $(D-1)n$ paths in $\mathcal{P}_n$. We finish the proof by lower bounding the number of vertices in $H_n$. The number of edges of $H_n$ in $V_i \times V_{i+1}$ is at least $Dn - |F| \geq (D-1)n$. In $\hat{H}_n$, each vertex in $V_i$ has exactly $D$ edges to $V_{i+1}$. Thus, at least $(D-1)n/D$ vertices of each part must remain in $H_n$ and so the number of vertices $|V(H_n)|$ satisfies $\frac{D-1}{D}nL_n \leq |V(H_n)| \leq nL_n$ This completes the proof of the claim.
\end{proof}
We obtain the family of base graphs $\{B_m\}$ as follows. We start with $H_n$ and add the following edges and vertices. For a vertex $v \in H_n$, let $\deg(v)$ be its degree. We may assume that $g_n$ is even. For each vertex $v \in V_1$, we add $\deg(v)$ edge-disjoint paths of $g_n/2$ edges from $s$ to $v$. For each vertex $v \in V_{L_n}$, we also add $\deg(v)$ edge-disjoint paths of $g_n/2$ edges from $v$ to $t$. Let $B_m$ be the resulting graph.

We now verify that $B_m$ satisfies the properties claimed in the lemma statement. 
Note that the total number edge-disjoint paths added from $s$ is equal to the number of edges going from $V_1$ to $V_2$, which equals $(D-1)n$, the number of edge-disjoint paths in $H_n$, and the same holds for $t$. 
The total number of vertices added is at most \[2\cdot ((D-1)n(g_n/2 - 1) + 1) = (D-1)n(g_n-2) + 2 \leq (D-1)n(g_n-1).\] \noindent so the number of vertices $m$ of $B_m$ satisfies 
\[m \leq (D-1)n(g_n-1) + nL_n = (D-1)n(g_n-1) + n(4g_n+1) \leq (D+3)ng_n.\] 
\noindent The number of edges in $B_m$ is 
\[(D-1)n(L_n - 1) + 2\cdot (D-1)ng_n/2 = 4(D-1)ng_n + (D-1)ng_n = 5(D-1)ng_n.\] 
Setting $D = 7$, we get that the number of edges of $B_m$ is at least $5(D-1)m/(D+3) = 3m$. 

 Moreover, $B_m$ can be decomposed into edge-disjoint $(s,t)$-paths such that each path is a shortest $(s,t)$-path in $B_m$. Thus, the graph $B_m$ satisfies the first three properties with width $\phi(m) = 2(g_n/2) + L_n - 1 = g_n + L_n -1 = \Theta(\log m)$. Since $H_n$ has girth at least $g_n$ and the additional edges cannot create a cycle of length smaller than $g_n$, we get that $B_m$ has girth $\Theta(\phi(m))$.\end{proof}

Next, given a deterministic embedding of a base graph $B_m$ given by Lemma~\ref{lem:base-HST} into a tree $T$, we consider the probability that a uniformly distributed edge of $B_m$ has high distortion. In fact, the following lemma considers a generalization in which the edges of the base graph are subdivided.

\begin{lemma}
  \label{lem:base-tree}
  For every positive integer $t \geq 1$, and some constant $c' > 0$, there exists an infinite family $\{B'_m(t)\}$ of graphs that satisfies the following properties:
  \begin{inparaenum}[\bfseries (1)]
  \item $B'_m(t)$ is a base graph with $m$ vertices; \item it has width $\phi'_m = \Theta(t\log (m/t))$; and
  \item \label{base-tree-dist} for any non-contractive embedding into a tree $T$, if we choose an edge $e$ in $B'_m(t)$ uniformly at random, then $\Pr[d_T(e) \geq c' \cdot\phi'_m] \geq 1/2t$.
  \end{inparaenum}
\end{lemma}

\begin{proof}
  Take $B_n$, for some $n$, as given by Lemma~\ref{lem:base-HST},
  and replace each edge with a path of $t$ edges. This results in a base graph $B'_m(t)$ on $m  = n + |E(B_n)|(t-1)$ vertices and $|E(B'_m(t))| = t|E(B_n)|$ edges. It also has width $\phi'_m = t\phi_n = \Theta(t\log (m/t))$ and girth $\Omega(t\log (m/t))$. Applying Lemma~\ref{lem:girth-dist} gives us 
\[
\frac{|E(B'_m(t))| - m}{|E(B'_m(t))|} = \frac{t|E(B_n)| - n - |E(B_n)|(t-1)}{t|E(B_n)|}  
=\frac{1}{t}\cdot\frac{|E(B_n)| - n}{|E(B_n)|} \geq \frac{2}{3t} > \frac{1}{2t},
\]
where the second-last inequality is due to property 3 of Lemma~\ref{lem:base-HST}.
\end{proof}

\subsection{Lower Bound for Tree Embeddings}
For each $i \leq \ell$, let $(u_i, v_i)$ be the random edge chosen at level $i$ of the recursive construction of $G_\ell$ and $(u_\ell, v_\ell)$ be the random vertex pair given by $\Dcal_\ell$.

\begin{lemma}
  \label{lem:tree}
  Consider the metric space $(V_\ell, d_{V_\ell})$ and distribution $\Dcal_\ell$ obtained from using $B'_m(t)$ as the base graph in the above framework. Then, for every deterministic online embedding of the sequence $\sigma$ of $\Dcal_\ell$ into a tree $T$, we have $E_{(\sigma, (u_\ell,v_\ell)) \in \Dcal_\ell}\left[\frac{d_T(u_\ell, v_\ell)}{d_{V_\ell}(u_\ell, v_\ell)}\right] \geq c'\phi'_m\cdot \frac{\ell}{2t}\left(1 - \frac{\ell}{2t}\right)$.
\end{lemma}

\begin{proof}
 Let $G_0, \ldots, G_\ell$ be the sequence of random graphs generated by using $B'_m(t)$ as the base graph in the framework. Recall that the sequence of terminals $\sigma$ consists of the vertices of $G_0$ first, then the new vertices of $G_1$, etc. For each $i \leq \ell$, let $(u_i, v_i)$ be the edge of $G_i$ that is replaced with a scaled-down copy of $B$ to form $G_{i+1}$. Define the random variable $c_i = \frac{d_T(u_i,v_i)}{d_{V_\ell}(u_i,v_i)}$ and let $C_i = E[c_i]$. We will need the following claim which says that conditioned on the distortion of $(u_i,v_i)$ being at least $\alpha$, the expected distortion of $(u_\ell,v_\ell)$ is also at least $\alpha$.

  \begin{claim}    
    For any $i \leq \ell$ and $\alpha$, we have $E\left[c_\ell \mid c_i \geq \alpha\right] \geq \alpha$.
  \end{claim}

  \begin{proof}
    We prove the claim by induction on $\ell$. The base case ($\ell = 0$) is vacuously true. We now prove the inductive case assuming that the statement holds for $\ell-1$. Fix $i < \ell$. The inductive hypothesis implies that $E\left[c_{\ell-1} \mid c_i \geq \alpha\right] \geq \alpha$. Let $A_{\ell-1}$ be the set of active edges of $G_{\ell-1}$. Now, $(u_\ell, v_\ell)$ is a uniformly random edge among $A_{\ell-1}$. Since $A_{\ell-1}$ is a scaled-down copy of $B$ with source $u_{\ell-1}$ and sink $v_{\ell-1}$, it consists of an edge-disjoint collection $\mathcal{P}_j$ of shortest $(u_{\ell-1},v_{\ell-1})$-paths. Thus, it suffices to prove that conditioned on $c_i \geq \alpha$, for each $P \in \mathcal{P}_j$, the expected average distortion of the edges in $P$ is at least $\alpha$, i.e.\ that $E\left[\frac{\sum_{e \in P} d_T(e)/(\phi'_m)^{-\ell}}{|P|} \mid c_i \geq \alpha\right] \geq \alpha$. Fix a path $P \in \mathcal{P}_j$. By triangle inequality and the inductive hypothesis, we have
    \begin{align*}
      E\left[\sum_{e \in P} d_T(e) \mid c_i \geq \alpha\right]
      \geq E[d_T(u_{\ell-1},v_{\ell-1}) \mid c_i \geq \alpha]
      \geq \alpha d_{V_\ell}(u_{\ell-1},v_{\ell-1})
      = \alpha |P|(\phi'_m)^{-\ell} 
    \end{align*}
    since $P$ is a shortest $(u_\ell,v_\ell)$-path. Thus, $E\left[c_\ell \mid c_i \geq \alpha\right] \geq \alpha$ as desired.
  \end{proof}
  Let $\event$ be the event that $c_i \geq c'\phi'_m$ for some $i \leq \ell$. The above claim implies that conditioned on $\event$, the expected distortion of $(u_\ell, v_\ell)$ is at least $c'\phi'_m$. Thus, it remains to show that the probability of $\event$ is at least $\frac{\ell}{2t}\left(1 - \frac{\ell}{2t}\right)$. Lemma~\ref{lem:base-tree} part \ref{base-tree-dist} implies that for each $i$, independent of the previous random choices, $\Pr[c_i \geq c'\phi'_m] \geq 1/2t$. Therefore
  \[\Pr[\event] \geq 1 - \left(1 - \frac{1}{2t}\right)^\ell \geq \frac{\ell}{2t} - \frac{\ell^2}{(2t)^2} = \frac{\ell}{2t}\left(1 - \frac{\ell}{2t}\right).\] 
    as desired.
  \end{proof}

We are now ready to prove Theorem~\ref{thm:dist-lb}.

\begin{proofof}{Theorem~\ref{thm:dist-lb}}
  By Yao's principle, Lemma~\ref{lem:tree} yields a sequence of $k = 2 + (m-2)\ell$ terminals such that any probabilistic online embedding into a tree has expected distortion at least $c'\phi'_m\cdot \frac{\ell}{2t}\left(1 - \frac{\ell}{2t}\right) \geq \Omega(\log \frac{m}{t} \cdot \ell \left(1 - \frac{\ell}{2t}\right))$. Now, the aspect ratio of these terminals is $\Phi_\ell = (\phi'_m)^\ell$ so $\ell = \frac{\log \Phi_\ell}{\log \phi'_m} = \Omega(\frac{\log \Phi_\ell}{\log t + \log \log (m/t)})$.   Setting $t = \ell$, we get that the expected distortion is at least
  \[ \Omega\left(\ell\log \frac{m}{\ell}\right) = \Omega\left(\frac{\log (m/\ell)\log \Phi_\ell}{\log \ell + \log \log (m/\ell)}\right) = \Omega\left(\frac{\log (m/\ell)\log \Phi_\ell}{\log \log \Phi_\ell + \log \log (m/\ell)}\right)\] where the last inequality is because $\ell \leq \log \Phi_\ell$. Now, setting $m = \Theta(k^{(1+\delta)/2})$ and $\ell = \Theta(k^{(1-\delta)/2})$, we get that $\log (m/\ell)  = \Theta(\delta \log k)$ and $\Phi_\ell = \Theta(\log k)^{\ell} = 2^{\Theta(\ell \log\log  k)} = 2^{\Theta(k^{(1-\delta')/2})}$ for any $\delta<\delta'<1$.
This completes the proof of the theorem.
\end{proofof}

\subsection{Lower bound for HST embeddings}
In the rest of this section, it will be convenient to use the labeled tree representation of HSTs. The main property of HST embeddings that our proof exploits is the following. 
\begin{prop}
  \label{prop:HST-cut}
  Let $G = (V,E)$ be an edge-weighted graph and consider a non-contractive embedding of its shortest-path metric $(V,d_V)$ into an HST $T$. Then, for every $u,v \in V$, the subset of edges $e \in E$ with $d_T(e) \geq d_T(u,v)$ is a cut-set for $u$ and $v$.
\end{prop}

\begin{proof}
  Let $L = d_T(u,v)$. The HST $T$ has a corresponding hierarchical partition of the metric space $(V, d_V)$. Let $P$ be the highest-level partition in which $u$ and $v$ are separated. We have that every edge $e \in E$ that is cut by $P$ has $d_T(e) \geq L$. Since $u$ and $v$ are separated by $P$, we get that every $(u,v)$-path must have at least one edge that is cut by $P$. Therefore, the subset of edges $e \in E$ with $d_T(e) \geq d_T(u,v)$ is a cut-set for $u$ and $v$. 
\end{proof}

\begin{lemma}
  \label{lem:HST}  
  Consider the metric space $(V_\ell, d_{V_\ell})$ and distribution $\Dcal_\ell$ obtained from using $B'_m(1)$ as a base graph in the above framework. Then, for every deterministic online embedding of the sequence $\sigma$ of $\Dcal_\ell$ into a HST $T$, we have
  $E_{(\sigma, (u_\ell,v_\ell)) \in \Dcal_\ell}\left[\frac{d_T(u_\ell, v_\ell)}{d_{V_\ell}(u_\ell, v_\ell)}\right] \geq \Omega(\ell\phi'_m) = \Omega(\ell\log m)$.
\end{lemma}

\begin{proof}
  Let $G_0, \ldots, G_\ell$ be the sequence of random graphs generated by using $B'_m(1)$ as the base graph in the framework. Recall that the sequence of terminals $\sigma$ consists of the vertices of $G_0$ first, then the new vertices of $G_1$, etc. For each $i \leq \ell$, let $e_i = (u_i, v_i)$ be the edge of $G_i$ that is replaced with a scaled-down copy of $B$ to form $G_{i+1}$.  Let  $A_i$ denote the support of $(u_i,v_i)$, i.e.~the active edges of $G_i$, and define $C_{i,j} = \{e \in A_i : d_T(e) \geq c' (\phi'_m)^{-j+1}\}$ where $c'$ is the constant in Lemma~\ref{lem:base-tree}.

  We begin with two observations. Firstly, $\Pr[e_j \in C_{j,j}] \geq 1/2$
by Lemma~\ref{lem:base-tree} part \ref{base-tree-dist}. Secondly, if $d_T(u_{i-1}, v_{i-1}) \geq (\phi'_m)^{-j+1}$ then  Proposition \ref{prop:HST-cut} implies that in the subgraph $A_i$, the set of edges $C_{i,j}$ is a cut-set for $u_{i-1}$ and $v_{i-1}$. Thus, conditioned on $d_T(u_{i-1}, v_{i-1}) \geq (\phi'_m)^{-j+1}$, the probability that a random edge of $A_i$ belongs to $C_{i,j}$ is at least the size of the min $(u_{i-1},v_{i-1})$-cut divided by $|E(A_i)|$. Since every $(u_{i-1},v_{i-1})$-path in $A_i$ has exactly $\phi'_m$ edges, we get that this ratio is at least $1/\phi'_m$ so
    $\Pr[e_i \in C_{i,j} \mid e_{i-1} \in C_{i-1,j}] \geq 1/\phi'_m$.

  We have 
    \[E[d_T(e_\ell)] = \sum_{j=1}^\ell E[d_T(e_\ell) \mid e_j \in C_{j,j}]\cdot \Pr[e_j \in C_{j,j}]\] and 
    \begin{align*}
    E[d_T(e_\ell) \mid e_j \in C_{j,j}]
    &\geq (\phi'_m)^{-j+1} \Pr[e_\ell \in C_{\ell,j} \mid e_j \in C_{j,j}] \\
    &\geq (\phi'_m)^{-j+1} \prod_{j < i \leq \ell} \Pr[e_i \in C_{i,j} \mid e_{i-1} \in C_{i-1,j}]\\
    &\geq (\phi'_m)^{-\ell+1}.
    \end{align*}

  Putting the above together, we get $E[d_T(e_\ell)] \geq \sum_{j = 1}^\ell (\phi'_m)^{-\ell+1} \cdot \frac{1}{2} = \frac{\ell(\phi'_m)}{2} d_V(e_\ell)$, as $d_V      (e_\ell) = (\phi'_m)^{-\ell}$.
\end{proof}
We are now ready to prove Theorem \ref{thm:dist-HST-lb}. 

\begin{proofof}{Theorem \ref{thm:dist-HST-lb}}
  By Yao's principle, Lemma~\ref{lem:HST} yields a sequence of $k = 2 + (m-2)\ell$ terminals such that any probabilistic online embedding into a HST has expected distortion at least $\Omega(\ell \log m)$. The aspect ratio of the terminals is $\Phi_\ell = (\phi'_m)^{\ell}$ so $\ell = \frac{\log \Phi_\ell}{\log \phi'_m} = \Omega(\frac{\log \Phi_\ell}{\log \log m})$ so the expected distortion is at least $\Omega\left(\frac{\log \Phi_\ell \log m}{\log \log m}\right)$. Setting $\ell = \Theta(k^{1-\delta})$ and $m = \Theta(k^\delta)$ gives the desired lower bound on expected distortion and aspect ratio $\Phi_\ell = \Theta(\log k)^{\ell} = 2^{\Theta(\ell \log\log  k)} = 2^{\Theta(k^{1-\delta'})}$ for any $\delta<\delta'<1$.
\end{proofof}

\section{A General Framework for Bypassing the Lower Bound}\label{sec:framework}
In this section we prove Theorem~\ref{thm:subadd-ub}. In the construction we use the property of abstract network design problems of admitting a min operator. We first prove the following theorem:

\begin{theorem}
  \label{thm:combine}
  Let $P$ be an Abstract Network Design problem with load function $\rho$. Suppose that $P$ admits a min operator, and suppose that there exists an algorithm $\Talg$ that is $\alpha$-competitive on instances defined on HST metrics. Let $\Balg$ denote a $\beta$-competitive algorithm (possibly randomized) on general metric inputs. Then,
 there exists a randomized algorithm $\combalg$ that, on every instance, has expected competitive ratio
  \[O(\alpha \cdot \min\{ \log k \cdot \log (\beta k \lambda_\rho), \log^2 k \cdot \log (\beta k)\} ),\]
  where $\lambda_\rho = \min\{\lceil\max_S \rho(S)/\min_{S \neq \emptyset}\rho(S)\rceil, r\}$. 

If the algorithms $\Talg$ and $\Balg$ have additive terms in their competitive ratio, denoted by $a$ and $b$ respectively, then $\combalg$ has an additive term $O(1)\cdot\max\{a, b\}$.
	
\end{theorem}
\noindent Since the dependency on the competitive ratio of the baseline algorithm decreases exponentially, this idea can be applied repeatedly, to completely remove the dependency on the original baseline algorithm as well. The additive term in the resulting algorithm will be $c^{\log^*(\beta)}$, for a constant $c=O(\max\{a,b\})$. This yields Theorem~\ref{thm:subadd-ub}.

\paragraph{Overview of $\combalg$.} We begin by describing the high-level intuition behind $\combalg$ and its analysis. Typically, tree embeddings are used to give a reduction from problems on general graphs to trees. This reduction consists of three steps: (1) compute a tree embedding $T$ of the input graph $G$; (2) solve the problem on $T$; (3) translate the solution on $T$ back to $G$. The analysis relies on the fact that the tree embedding is non-contractive to argue that the cost of the translated solution is at most the cost of the tree solution. Then, it bounds the cost of the optimal solution in $T$ in terms of the optimal solution in $G$. This is done by considering the translation of the optimal solution in $G$ into $T$ and upper bounding the blow up in cost by the distortion of the embedding.

Our approach uses the combining scheme of the min operator property to combine the usual tree-embedding-based algorithm (using the embedding of Theorem~\ref{thm:summary}) and the baseline algorithm. At a high level, this allows us to fall back on the baseline algorithm in the case that the optimal solution in $T$ becomes very expensive due to some edges of $G$ being distorted badly. In the analysis, we first show (Lemma~\ref{lem:cost}) that the resulting algorithm has cost at most
$O(\alpha) \min\{\OPT(T), \beta\OPT\}$, where $\OPT(T)$ is the cost of the optimal solution in $T$ and $\OPT$ is the cost of the optimal solution in $G$. We then bound $\min\{\OPT(T),\beta\OPT\}$ by considering the translation of the optimal solution in $G$ into $T$ using a more refined analysis to bound the overhead of the translation.

At a high level, the embedding of Theorem~\ref{thm:summary} uses a collection of $O(\log \Phi)$ probabilistic partitions, one per scale. Each partition contributes a $O(\log k)$ factor to the distortion, resulting in an overall distortion of $O(\log k \log \Phi)$. The key idea is that to bound $\min\{\OPT(T), \beta \OPT\}$, it suffices to only consider much fewer than $O(\log \Phi)$ scales. It is straightforward to see that we can ignore scales above $\beta\OPT$. Arguing that scales much smaller than $\OPT$ do not contribute much to the cost of the translated solution is more difficult. If there exists a near-optimal solution in $G$ that uses at most $\poly(k)$ edges, then we can ignore scales that are smaller than $\OPT/\poly(k)$ and so only $\log (\beta k)$ scales are relevant. However, in general it is unclear that such a bound is possible. Instead, we use a more subtle argument (Lemma~\ref{lem:translation-cost})
based on a notion of sparsity that we call ``tree sparsity'' (Definition~\ref{def:sparsity}). We apply this argument in two different ways: one based on the parameter $\lambda_\rho$ and the fact that $T$ has $O(k)$ edges; the other based on the fact that one can use offline tree embeddings to obtain a $O(\log k)$-approximate solution in $G$ that uses at most $O(k)$ edges. The minimum of the resulting bounds yields the bound in Theorem~\ref{thm:combine}. 

{\paragraph{Description of $\combalg$.} Let $\Talg$ be an algorithm that is $\alpha$-competitive with an additive term $a$ on HST instances,  and $\Balg$ to be a baseline algorithm that is $\beta$-competitive with an additive term $b$ on arbitrary instances. Let $\Tembed$ be the online embedding algorithm of Theorem~\ref{thm:summary} and $\alg$ be the randomized algorithm of Theorem~\ref{thm:reduction-trees} that uses $\Tembed$ as its online tree embedding algorithm. The algorithm $\combalg$ is the combination of $\Balg$ and $\alg$.

\subsection{Analysis}

Let $\alg[T]$ denote the algorithm $\alg$ with the fixed choice of $\mu$-HST metric $(T, d_T)$ used by $\Tembed$, and let $\combalg[T]$ denote the combining algorithm of $\alg[T]$ and $\Balg$, as in Definition ~\ref{def:min_operator-inf}.
Observe that for the problem on $T$, the subgraphs of any feasible solution must be a subgraph of $T$. 

\begin{lemma}\label{lem:cost}
We have $\cost(\combalg[T])\leq O(\alpha)\min\{\OPT(T), \beta \OPT\} +O(1)\max\{a, b\}$.

\noindent
When $\Talg$ and/or $\Balg$ are randomized algorithms, then the expected value of the cost of the $\combalg[T]$ is bounded by the above,  
where the expectation is over the internal randomness of $\Talg$ and/or $\Balg$ (but not over the random choice of $T$).
\end{lemma}

\begin{proof}
By the definition, 
\begin{align*}
    E[\cost(\combalg[T])] 
    &\leq O(1) E[\min\{ \cost(\alg[T]), \cost(\Balg)\}] \\
    &\leq O(1) \min\{E[\cost(\alg[T])], E[\cost(\Balg)]\} \\
    &\leq O(1) \min \{E[\cost(\alg[T])], \beta \OPT +b\}.
\end{align*}
By Claim~\ref{clm:1} it holds that 
$E[\cost(\alg[T])] \leq E[\cost(\Talg[T])] \leq \alpha \OPT(T)+a$. Therefore, 
\begin{align*}
    E[\cost(\combalg[T])] 
    &\leq O(1)\min\{\alpha \OPT(T)+a, \beta \OPT +b\} \\
    &\leq O(1) \min\{\alpha \OPT(T), \beta \OPT\} +O(1)\max\{a, b\},
\end{align*} which completes the proof. 
\end{proof}

We now bound $\OPT(T)$ in terms of any feasible solution $\Scal$ in $G$. Let $\R_i = (R_i, C_i)$ be its $i$-th response. Consider the solution $\Scal'$ in $T$ obtained by translating $\Scal$ into $T$ using the extension function $F := F_r$ as follows. For an edge $e = (u,v) \in G$, let $I_e = \{i : e \in R_i\}$ and $P_T(e)$ denote the path in $T$ between $F(u)$ and $F(v)$. The $i$-th response of $\Scal'$ is $(R'_i, C_i)$ where $R'_i = \bigcup_{e \in R_i} P_T(e)$. Since any terminal pair that is connected in $R_i$ is also connected in $R'_i$, the solution $\Scal'$ is a feasible solution on $T$. Thus, $\OPT(T) \leq \cost(\Scal')$. Since each tree edge $e' \in T$ is used by $\Scal'$ in the time steps $i$ such that $R_i$ contains an edge $e \in G$ with $e' \in P_T(e)$, we have
\begin{equation}
  \OPT(T) \leq \sum_{e' \in T} d_T(e') \rho\left(\bigcup_{e \in G: e' \in P_T(e)} I_e\right).\label{eq:OPT(T)}
\end{equation}
While one can use subadditivity to upper bound the RHS of the inequality by $\sum_{e \in G} d_T(e) \rho(I_e)$, we prove a more refined bound based on the following notion of sparsity.

\begin{definition}[Tree sparsity]
  \label{def:sparsity}
 Consider a solution $\Scal$ in $G$. For a HST embedding $T$ of $G$, let $\Ehat_T$ be the smallest edge subset of $G$ such that $\sum_{e' \in T} d_T(e') \rho\left(\bigcup_{e \in G: e' \in P_T(e)} I_e\right) \leq \sum_{e \in \Ehat_T} d_T(e)\rho(I_e)$. The \emph{tree sparsity} of $\Scal$ is the maximum size of $\Ehat_T$ over non-contracting HST embeddings of $G$. 
\end{definition}
Intuitively, if a solution $\Scal$ has small tree sparsity, then for any HST embedding $T$, the cost of its translation into $T$ can be charged to a small subset of edges in $G$. Note that subadditivity of $\rho$ immediately implies that $|\cup_i R_i|$ is an upper bound on the tree sparsity of $\Scal$ .

\begin{lemma}
  \label{lem:translation-cost}
  Let $\Scal$ be a solution in $G$ with tree sparsity $\eta$. Then,
  $E_T[\min\{\OPT(T), \beta \OPT\}] \leq O(\log k \log (\beta \eta))\cost(\Scal)$.
\end{lemma}

We will need the following claim in the proof of the lemma. The claim bounds the expected value of the truncated distance on the random tree $T$ of the embedding. Particularly, for constants $\sigma$, $\gamma$ and $\delta$ the random variable $\min\{\sigma d_T(e), \gamma\}$ is considered, in the distance scales at least $\delta$.
\begin{claim}
  \label{claim:tdist}
  Let $(V,d)$ be a metric space and $t_1, \ldots, t_k$ be a sequence of $k$ terminals. Consider the $\mu$-HST embedding $T$ of Theorem~\ref{thm:summary}. Then, for any $e=(u,v)$ with $u,v \in V$ and $\gamma,\sigma,\delta > 0$ such that $\gamma > \delta$, the  random variable
  \[\dhat_T(e) =
    \begin{cases}
      \gamma & \mbox{for } \sigma d_T(e) \geq \gamma \\
      \sigma  d_T(e) & \mbox{for } \delta \leq \sigma  d_T(e) < \gamma \\
      0 & \mbox{for } \sigma  d_T(e) < \delta
    \end{cases}
  \]
  has expectation (over the choice of $T$) $E_T[\dhat_T(e)] \leq O(\log k \cdot \log \frac{\gamma}{\delta}) \cdot \sigma d(e)$.
\end{claim}

\begin{proof}
  For an event $\mathcal{E}$, let $1\{\mathcal{E}\}$ be its indicator variable. Since $T$ is a $\mu$-HST, we have
  \[\dhat_T(e) \leq \gamma \cdot 1\{d_T(e) \geq \gamma/\sigma\} + O(1) \sum_{j : \mu^j \in [\delta, \gamma]} \mu^j \cdot 1\{d_T(e) \geq \mu^j/\sigma\}.\]
and so
\[  E_T[\dhat_T(e)]
  \leq \gamma \cdot \Pr_T[d_T(e) \geq \gamma/\sigma] + O(1) \sum_{j : \mu^j \in [\delta, \gamma]} \mu^j \cdot \Pr_T[d_T(e) \geq \mu^j/\sigma].\]
By Theorem~\ref{thm:summary}, $\Pr_T[d_T(e) \geq \mu^j/\sigma] \leq O(\log k)\frac{d(e)}{\mu^j/\sigma}$ for each $j$ and $\Pr_T[d_T(e) \geq \gamma/\sigma] \leq O(\log k)\frac{d(e)}{\gamma/\sigma}$. As there are at most $O(\log \frac{\gamma}{\delta})$ terms in the sum, we have $E_T[\dhat_T(e)] \leq O(\log k \cdot \log \frac{\gamma}{\delta})\cdot \sigma d(e)$, as desired.
\end{proof}

\begin{proof}[Proof of Lemma \ref{lem:translation-cost}]
   Fix an HST embedding $T$. Let $\phi_T(e) = d_T(e) \rho(I_e)$. Since $\Scal$ has tree sparsity $\eta$, there exists $\Ehat_T \subseteq G$ of size at most $\eta$ with $\sum_{e' \in T} d_T(e') \rho\left(\bigcup_{e \in G: e' \in P_T(e)} I_e\right) \leq \sum_{e \in \Ehat_T} \phi_T(e)$. Using Inequality~\eqref{eq:OPT(T)}, we get 
\[    \OPT(T)
    \leq \sum_{e \in \Ehat_T} \phi_T(e) 
    \leq \sum_{e \in \Ehat_T : \phi_T(e) \geq \frac{\OPT}{\eta}} \phi_T(e) + \OPT 
    \leq \sum_{e \in G : \phi_T(e) \geq \frac{\OPT}{\eta}}  \phi_T(e) + \OPT,\]
  where the last inequality uses the fact that $\Ehat_T$ is a subset of edges in $G$. 
  Thus, it suffices to bound 
  $E_T[\min\{\sum_{e \in G : \phi_T(e) \geq \frac{\OPT}{\eta}}  \phi_T(e), \beta\OPT\}]$.
  Now, observe that
  \begin{align*}
  \min\left\{\sum_{e \in G : \phi_T(e) \geq \frac{\OPT}{\eta}} \phi_T(e), \beta\OPT\right\} 
   &\leq \sum_{e \in G : \phi_T(e) \geq \frac{\OPT}{\eta}} \min\{ \phi_T(e), \beta\OPT\} \\
    &= \sum_{e \in G} 1\{\phi_T(e) \geq \frac{\OPT}{\eta}\} \min\{ \phi_T(e), \beta\OPT\},
  \end{align*}
  where $1\{\phi_T(e) \geq \frac{\OPT}{\eta}\}$ is the indicator variable for the event $\phi_T(e) \geq \frac{\OPT}{\eta}$. Consider the random variable $\dhat_T(e)$ where
  \begin{align*}
   \dhat_T(e) = 
    \begin{cases}
      \beta \OPT & \mbox{for } d_T(e) \rho(I_e) \geq \beta\OPT \\
      d_T(e) \rho(I_e) & \mbox{for } \frac{\OPT}{\eta} \leq  d_T(e) \rho(I_e) < \beta\OPT \\
      0 & \mbox{for }  d_T(e)\rho(I_e) < \frac{\OPT}{\eta}
    \end{cases}
  \end{align*}
  We have that \[\dhat_T(e) = 1\{\phi_T(e) \geq \frac{\OPT}{\eta}\} \min\{ \phi_T(e), \beta\OPT\}.\] So,
    \[E_T\left[\min\left\{\sum_{e \in G : \phi_T(e) \geq \frac{\OPT}{\eta}}  \phi_T(e), \beta\OPT\right\}\right]
    \leq \sum_{e \in G} E_T[\dhat_T(e)] \leq O(\log k \cdot \log (\beta \eta))\sum_{e \in G} d(e)\rho(I_e),\]
  where the last inequality follows from applying Proposition~\ref{claim:tdist} to each $e \in G$ with $\gamma = \beta \OPT, \delta = \OPT/\eta, \sigma = \rho(I_e)$. Since $\cost(\Scal) = \sum_{e \in G} d(e)\rho(I_e)$, this completes the proof of the lemma.
\end{proof}

With Lemma~\ref{lem:translation-cost} in hand, to complete the proof of Theorem~\ref{thm:combine}, it remains to show that: (1) any optimal solution has tree sparsity $O(k\lambda_\rho)$ (Lemma~\ref{lem:OPT-T}); (2) there exists an $O(\log k)$-approximate solution with tree sparsity $k-1$ (Lemma~\ref{lem:OPT-embed}). We first show (1). 
\begin{prop}
  \label{prop:subadd-ar}
  For any collection $\Acal$ of nonempty subsets of $\{1, \ldots, r\}$, there exists a subcollection $\Acalhat \subseteq \Acal$ of size at most $\lambda_\rho$ such that $\rho\left(\bigcup_{A \in \Acal} A\right) \leq \sum_{A \in \Acalhat} \rho(A)$.
\end{prop}

\begin{proof}
Case $1$: $\lambda_{\rho} = \lceil\fmax/\fmin\rceil$. 
If $|\Acal| \leq  \lambda_{\rho}$, the proposition follows trivially by taking $\Acalhat=\Acal$ and using the subadditivity of $\rho$.   Suppose that $|\Acal| > \lambda_{\rho}$, let $\Acalhat$ be any subcollection of $\Acal$ of size $\lambda_{\rho}$. Observe that
    $\sum_{A \in \Acalhat}\rho(A) \geq \lambda_{\rho} \min_{S \neq \emptyset} \rho(S) \cdot  \geq \max_S \rho(S) \geq \rho\left(\bigcup_{A \in \Acal} A\right)$.
    Case $2$: $\lambda_{\rho}=r$. Construct $\Acalhat$ by greedily choosing sets from $\Acal$, iteratively choosing a set which adds at least one new element of $\{1, \ldots, r\}$, until the union of sets of $\Acalhat$ equals that of $\Acal$. Then $|\Acalhat| \leq r$ and the result follows from subadditivity of $\rho$.
  \end{proof}

\begin{lemma}
  \label{lem:OPT-T}
  Any optimal solution has tree sparsity $O(k\lambda_\rho)$.
\end{lemma}

\begin{proof}
  Let $\Scal$ be an optimal solution and $T$ be a HST embedding. For each $e' \in T$, let $\Ehat_T(e')$ be a subset of at most $\lambda_\rho$ edges $e \in G$ with $e' \in P_T(e)$ such that $\rho\left(\bigcup_{e \in G: e' \in P_T(e)} I_e\right) \leq \sum_{e \in \Ehat_T(e')}\rho(I_e)$ (such a subset exists by Proposition~\ref{prop:subadd-ar}). Let $\Ehat_T = \bigcup_{e' \in T} \Ehat_T(e')$; note that $|\Ehat_T| \leq O(k\lambda_\rho)$ since $T$ has at most $O(k)$ edges (by Theorem~\ref{thm:summary}). So, we have
  \begin{align*}
   \sum_{e' \in T}d_T(e') \rho\left(\bigcup_{e \in G: e' \in P_T(e)} I_e\right) 
  &\leq \sum_{e' \in T}d_T(e')\sum_{e \in \Ehat_T(e')}\rho(I_e) \\
  &= \sum_{e \in \Ehat_T} \left(\sum_{e' \in T : e \in \Ehat_T(e')} d_T(e')\right) \rho(I_e) \\
  &\leq \sum_{e\in \Ehat_T}d_T(e) \rho(I_e)
\end{align*}
where the last inequality follows from the fact that the set of tree edges $e'$ with $e \in \Ehat_T(e')$ is a subset of the path in $T$ between the endpoints of $e$, and the total length of the path in $T$ is exactly $d_T(e)$. 
\end{proof}

Finally, we prove (2). We will need the following observation from Gupta, Nagarajan, and Ravi~\cite[Theorem 7]{GuptaNR10}, based on the probabilistic embedding of \cite{FRT04,Bartal04}, and a claim on Steiner points removal in HST trees (see e.g. Theorem 5.1 in \cite{KonjevodRS01} or the construction of the extension $H_i$ in the proof of Claim~\ref{extend_HST} in Appendix~\ref{online_embedding}).

\begin{theorem}
  \label{thm:GNR}
  Let $(M,d_M)$ be a metric space with a designated subset $W \subseteq M$. Then there is a distribution $\mathcal{T}$ of HSTs with leaves $M$ such that for every $T \in \mathcal{T}$, we have $d_T(x,y) \geq d_M(x,y)$ for every $x,y \in W$ and $E_{T \sim \mathcal{T}}[d_T(u,v)] \leq O(\log |W|)d_M(u,v)$ for every $u,v \in M$.
\end{theorem}

\begin{claim}
  \label{claim:KRS}
  Let $(T, d_T)$ be a HST metric and $Z$ be a subset of its vertices. Then there exists an embedding $g$ of $T$ into a tree metric $(T', d_{T'})$ whose vertex set is exactly $Z$ such that $d_{T'}(g(u),g(v)) \leq d_T(u,v)$ for every $u, v \in T$ and $d_{T}(x,y) \leq 4d_{T'}(g(x),g(y))$ for every $x,y \in Z$.
\end{claim}

\begin{lemma}
  \label{lem:OPT-embed}
  There exists a $O(\log k)$-approximate solution with tree sparsity $k-1$.
\end{lemma}

\begin{proof}
  Theorem~\ref{thm:GNR} implies that there exists a HST embedding $T$ such that $\OPT(T) \leq O(\log k)\OPT$ and $d_T(u,v) \geq d(u,v)$ for every pair of terminals $u,v \in \Zcal$. Let $T'$ and $g$ be the tree metric and the embedding guaranteed by applying Claim~\ref{claim:KRS} to $\Zcal$ and the subtree of $T$ induced by $\Zcal$. Let $\OPT(T')$ be the cost of the optimal solution for the instance on $T'$ induced by $g$. We now argue that any feasible solution on $T$ can be transformed into a feasible solution on $T'$ of less or equal cost. Let $\Scal = ((R_1, C_1), \ldots, (R_r, C_r))$ be a feasible solution on $T$ and $\Scal' = ((R'_1, C'_1), \ldots, (R'_r, C'_r))$ with $R'_i = g(R_i)$ and $C'_i = g(C_i)$. The solution $\Scal'$ is feasible for the instance on $T'$ induced by $g$, and $\cost(\Scal') \leq \cost(\Scal)$. Thus, $\OPT(T') \leq \OPT(T) \leq O(\log k)\OPT$.

  Let $\Scal^*$ be the optimal solution for $T'$ and $R^*_1, \ldots, R^*_r$ be the subgraphs it uses. Since the vertices of $T'$ is a subset of $G$ (in particular, it is exactly the set of terminals $\Zcal$), $\Scal^*$ is also a feasible solution in $G$. Moreover, since $d(u,v) \leq d_T(f(u),f(v)) \leq 4d_{T'}(g(f(u)), g(f(v)))$ for every $u,v \in \Zcal$, the cost of $\Scal^*$ in $G$ is at most $4\OPT(T') \leq O(\log k)\OPT$. Finally, the tree sparsity of $\Scal^*$ is at most $|\cup_i R'_i| \leq k-1$ since $T'$ is a tree with $k$ vertices.
\end{proof}

Combining Lemma~\ref{lem:translation-cost} with Lemmas~\ref{lem:OPT-T} and \ref{lem:OPT-embed} gives us Theorem~\ref{thm:combine}.

\section*{Acknowledgements}
We thank Arnold Filtser for pointing out an error in an earlier proof of Lemma~\ref{lem:base-HST} and suggesting proving a direct construction.

\bibliographystyle{alpha}{\small \bibliography{onlineemb,art}}

\appendix
\section{Metric Completion of an Input Graph}\label{metric_completion}
\begin{claim}
Let $P$ be an Abstract Network Design problem and let $G$ be its input graph. Then, w.l.o.g. $G$ is a complete graph with edge length satisfying triangle inequality. 
\end{claim}
\begin{proof}
 We construct a complete graph $\hat{G}$ as follows: remove all the edges in $G$ that form triangles yet do not satisfy the triangle inequality, and complete the graph by adding edges $e=(u,v)$ with $d(e)$ being the length of a shortest path in $G$ between $u$ and $v$. 

Let $\hat{\mathcal{S}}_i=((\hat{{R}}_1, \hat{C}_1), \ldots, (\hat{{R}}_i, \hat{C_i}))$ denote a feasible solution to $P$ with the input graph being $\hat{G}$. We show that there is a feasible solution $\mathcal{S}_i=(({R}_1, C_1) \ldots, ({R}_i, C_i))$ for $P$ with the input graph being $G$, such that $\cost(\mathcal{S}_i) \leq \cost(\hat{\mathcal{S}}_i)$. For each $\hat{R}_j \in \hat{\mathcal{S}}_i$, for all $e=(u,v) \in \hat{R}_j$ that either has been added to $\hat{R}_j$ or the weight of which has been updated in $\hat{R}_j$ let $p_G(e)$ denote a shortest path between $u$ and $v$ in $G$. We add the edges of the path $p_G(e)$ to the response $R_j$. All the other edges in $\hat{R}_j$ are added to $R_j$ as well. The connectivity lists $C_j$ are defined to be exactly the lists $\hat{C}_j$. 

By the construction, the solution $\mathcal{S}_i$ is feasible.   
In addition, if $e' \in \hat{\mathcal{R}}_j$ is such that $e' \notin R_j$ or such that $d_G(e') \neq d_{\hat{G}}(e')$, then by the construction of $\hat{G}$, $d_{\hat{G}}(e')$ equals to the length $p_G(e') \subseteq R_j$. Thus, by the subadditivity of the load function we conclude that $\cost(\mathcal{S}_i) \leq \cost(\hat{\mathcal{S}}_i)$. 

On the other hand, let $\mathcal{S}^*_i$ be an optimal solution on the input graph $G$. Note that since its response subgraphs $R^*_j$ do not use edges that violate triangle inequality, the solution $\mathcal{S}^*_i$ is also a feasible solution on $\hat{G}$, of the same cost. This implies that the costs of an optimal solution on $G$ and on $\hat{G}$ are equal.  This completes the proof.
\end{proof}

\section{$s$-Server Problem}\label{app:server}

We prove that without loss of generality we can assume that at the beginning all the servers are located at the same point.
If the servers are located in different locations, then in the representation of the problem in terms of Abstract Network Design, the size of the terminal set is $k=r+s$. If $s=O(r)$, then $k=O(r)$ and the competitive ratio stated in Corollary ~\ref{cor_server} is also true for this case. 

Thus, we assume that $r=O(s)$,  implying $k=r+s=O(s)$. Therefore, the competitive ratio obtained by Corollary ~\ref{cor_server} is $O(\log^4 s)$. Let $A$ be an algorithm that solves the $s$-server problem, on any general metric input space, with the assumption that all the serves at the beginning are located at the same point $v_0 \in V$. Let $\gamma=O(\log^4 s)$ and $a$ be such that 
$\cost(A) \leq \gamma \cdot \OPT +a$, where $\OPT$ is the cost of the optimal solution. Consider the algorithm $\hat{A}$ that given an instance where the servers are not necessarily located at the same point first moves all the servers to $v_0$ and then runs the algorithm $A$. It holds that $\cost(\hat{A}) \leq \cost(A)+b$, where $b$ is some constant (dependent on the locations of the servers). Then, $\cost(\hat{A}) \leq \gamma \cdot \OPT + a +b$. Let $\hat{\OPT}$ be the cost of the optimum algorithm for the instance where servers are not located at the same point at the beginning. Therefore, $\OPT \leq {\hat\OPT} + c$, where $c$ is a constant dependent on the locations of the servers. We have, $\cost(\hat{A}) \leq \gamma \cdot \hat{\OPT} + \gamma \cdot c + a + b$. Since $\gamma=O(\log^4 s)$ is a constant (in terms of the parameters of the problem $r$ and $k$), this implies that the competitive ratio of $\hat{A}$ is the same as the competitive ratio of $A$. This completes the poof.

 \section{Min Operator For Abstract Network Design}\label{onlineMin}

Let $P$ be an Abstract Network Design problem, and let $A$ be an online deterministic algorithm solving $P$. For a sequence of $i \geq 1$ requests $(Z_1, \ldots, Z_i)$, let $\mathcal{S}^A_i=(\mathcal{R}^A_1, \ldots, \mathcal{R}^A_i)$ denote the feasible solution of $A$. 

\begin{definition}\label{def:min_operator}[Restatement of Definition~\ref{def:min_operator-inf}]
We say that an online problem $P$ \textsl{admits a min operator} with factor $\eta  \geq 1$ if 
for any given deterministic online algorithms $A$ and $B$, there is an online deterministic algorithm $C$, such that for any $(Z_1, \ldots, Z_r)$ requests to $P$,
 $\cost(\mathcal{S}^C_r) \leq \eta \cdot \min\{\cost(\mathcal{S}^A_r), \cost(\mathcal{S}^B_r) \}$.
\end{definition}

Note that if either $A$ and/or $B$ are randomized, it follows from the definition that $C$ is randomized and $E[\cost(C)] \leq \eta \cdot E[\min\{\cost(A), \cost(B)\}] \leq \eta \cdot \min \{ E[\cost(A)], E[\cost(B)]\}$. Therefore, we focus on the case of the deterministic algorithms in the min operator.

\begin{definition}[Sequencing property]\label{def:seq:prop}
We say that a problem $P$ \emph{satisfies the sequencing property with parameters $(a, b)$}, $a, b \geq 1$, if for any (possibly  online) deterministic algorithms $\hat{A}$ and $\hat{B}$, for any $(Z_1, \ldots, Z_i)$ requests and for any sequence of new $t-i$ requests $(Z_{i+1}, \ldots, Z_t)$, for any $t>i$, there exists a sequence of responses $((R'_{i+1}, C'_{i+1}),\ldots, (R'_{t}, C'_t))$ such that 
\begin{enumerate}
\item $\mathcal{F}_t(C^{\hat{A}}_1, \ldots, C^{\hat{A}}_i, C'_{i+1}, \ldots, C'_{t})=1$;
\item $\cost(\mathcal{R}^{\hat{A}}_1, \ldots, \mathcal{R}^{\hat{A}}_i, \mathcal{R}'_{i+1}, \ldots, \mathcal{R}'_{t}) \leq a \cdot \cost(\mathcal{R}^{\hat{A}}_1, \ldots, \mathcal{R}^{\hat{A}}_{i}) + b \cdot \cost(\mathcal{R}^{\hat{B}}_1, \ldots, \mathcal{R}^{\hat{B}}_t)$.
\end{enumerate}
\end{definition}

We next show that the sequencing property is enough to ensure the problem admits min operator. 

\begin{lemma}\label{min_general_lemma}
If an Abstract Network Design problem $P$ satisfies the sequencing property with parameters $(a, b)$, then $P$ admits min operator with factor $\eta=O(a^{3.5} \cdot b)$.
\end{lemma}

\begin{proof}

Let $A$ and $B$ be any deterministic online algorithms for $P$. For a sequence of requests $(Z_1, \ldots, Z_t)$ let $\cost_{A}(t)$ (let $\cost_B(t)$) denote the cost of the solution of the algorithm $A$ (the algorithm $B$) on the sequence $(Z_1, \ldots, Z_t)$. Let $M$ denote $\min\{A, B\}$ (i.e., the algorithm of smaller cost).

The algorithm $C$ proceeds in phases. 
Let $t_1>1$ be the maximal index such that for the sequence of requests $(Z_1, \ldots, Z_{t_1})$ it holds that $\cost_{A}({t_1}) \leq \gamma \cdot \cost_{B}({t_1})$ for some $\gamma \geq 1$ that will be chosen later, and let $t_2'=t_1+1$. Note that $\cost_{A}(t_2')>\gamma \cdot \cost_{B}(t_2')$ by the definition of $t_1$. We say that the first phase occurs at the timesteps $[1, t_1]$. The responses of $C$ during the first phase are the responses of $A$. Let $C_1$ denote the algorithm $C$ in the first phase.

For $i>1$ the phase $i$ is defined by the following: the phase occurs at the timesteps $[t_{i}', t_i]$, where $t_i>t_{i}'$ is defined as the maximal index for which $\cost_{B}(t_i) \leq \gamma \cdot \cost_{A}(t_i)$ if $i$ is even, and $t_i>t_{i}'$ is the maximal index for which $\cost_{A}(t_i) \leq \gamma \cdot \cost_{B}(t_i)$ if $i$ is odd.  
Let $\hat{A}_i=C_{i-1}$; let $\hat{B}_i=B$ if $i$ is even and $\hat{B}_i=A$ if $i$ is odd. Since $P$ satisfies sequencing property it holds that for $\hat{A}$ and $\hat{B}$, for any sequence of new requests $(Z_{t_{i}'}, \ldots, Z_{t_j})$, for any $t_j \geq t_i'$, there are responses $((R'_{t_{i}'} ,C'_{t_{i}'}), \ldots, (R'_{t_j}, C'_{t_j}))$ such that $\cost(\mathcal{R}^{\hat{A}}_{1}, \ldots, \mathcal{R}^{\hat{A}}_{t_{i-1}}, \mathcal{R}'_{t_{i}'}, \ldots, \mathcal{R}'_{t_j}) \leq a \cdot \cost_{\hat{A}}(t_{i-1})+b \cdot \cost_{\hat{B}}(t_j)$. 
We define algorithm  $C_i$ to be the algorithm that outputs the responses of $C_{i-1}$ followed by the responses obtained from the sequencing algorithm $\hat{A}$ with $\hat{B}$, i.e., the responses $\{(R'_{t'_l},C'_{t'_l} )\}_{l\geq 1}$. 

We prove the following claim:

\begin{claim}
Let $\zeta=4a^{1.5}(a^2+a+1)b$ and $\gamma=2a^{1.5}$. For all $i\geq 2$, for all $t \in [t_i', t_i]$ it holds that $\cost_{C_i}(t) \leq \zeta \cdot \cost_{M}(t)$.

\end{claim} 

\begin{proof}
The proof is by induction on the phase $i\geq 3$. Let $i\geq 3$ and $t\in [t_i', t_i]$. By the definition of the algorithm $C_i$ and by the sequencing property, we have $\cost_{C_i}(t) \leq a \cdot \cost_{C_{i-1}}(t_{i-1})+b \cdot \cost_{\hat{B}_i}(t) \leq a \cdot \cost_{C_{i-1}}(t_{i-1}) + b \cdot \gamma \cdot \cost_{M}(t)$, where the second inequality holds by definition of $\hat{B}_i$.
Applying this rule on $\cost_{C_{i-1}}(t_{i-1})$ for two more steps, we obtain $\cost_{C_i}(t) \leq a^3\cdot \cost_{C_{i-3}}(t_{i-3})+a^2(a+1)b\gamma \cdot\cost_{M}(t)$. By induction assumption it holds that $\cost_{C_{i-3}}(t_{i-3})\leq \zeta \cdot \cost_{M}(t_{i-3})$. Bellow we show that for all $t> t_i'$, $\cost_{M}(t) \geq \gamma^2 \cdot \cost_{M}(t_{i-3})$. Thus, we conclude: $\cost_{C_i}(t)\leq (a^3\zeta/\gamma^2 + a^2(a+1)b \gamma)\cdot \cost_{M}(t) \leq \zeta \cost_{M}(t)$ for chosen values of $\zeta$ and $\gamma$.

It remains to show that for all $t> t_i'$, $\cost_{M}(t) \geq \gamma^2 \cdot \cost_{M}(t_{i-3})$. By the monotonicity of the load function, it is enough to prove that $\cost_{M}(t'_i) \geq \gamma^2 \cost_{M}(t'_{i-2})$. Assume w.l.o.g. that at the phase $i$ it holds that $\cost_{B}(t) \leq \gamma \cost_{A}(t)$ for all timesteps $t$ in the phase (i.e., assume that $i$ is an even phase). Then, by the definition of $t_i'$, it holds that $\cost_{A}(t_i') > \gamma \cost_{B}(t'_i)$, i.e., $M(t'_i)=B$. Similarly, by the definition of 
$t'_{i-2}$, we conclude that $M(t'_{i-2})=B$. Since $\cost_{M}(t)=\cost_{B}(t) \geq \cost_{B}(t'_{i-1})>\gamma \cost_{A}(t_{i-1}')\geq \gamma \cost_A(t'_{i-2}) > \gamma^2 \cost_B(t'_{i-2}) \geq  \gamma^2 \cost_{B}(t_{i-3})$, which completes the proof. 
\end{proof}

\end{proof} This completes the proof of the lemma.

\begin{corollary}\label{cor_min_operator}
If the feasibility requirement of an Abstract Network Design problem $P$ is memoryless then $P$ admits a min operator.
\end{corollary}

\begin{proof}
By Lemma~\ref{min_general_lemma}, it is enough to show that $P$ satisfies the sequencing property. Assume we are given two deterministic online algorithms $A$ and $B$. For any $t>i \geq 1$, for a sequence of new requests $(Z_{i+1}, \ldots, Z_t)$ let $((R^B_{i+1},C^B_{i+1}), \ldots, (R^B_{t}, C^B_{t}))$ be the feasible responses of the algorithm $B$. Then, define $(R'_{i+1}, C'_{i+1} )=(R^B_{i+1}, C^B_{i+1}), \ldots, (R'_{t}, C'_{t})=(R^B_{t},C^B_{t})$. Since the feasibility function is memoryless $\mathcal{F}_{i+j}(C^A_{1}, \ldots, C^A_{i}, C'_{i+1}, \ldots, C'_{i+j})=1$, for all $ 1\leq j\leq t-i$. In addition, by the subadditivity of the load function it holds that
\begin{align*}
  \cost(\mathcal{R}^A_{1}, \ldots, \mathcal{R}^A_{i}, \mathcal{R}'_{i+1}, \ldots, \mathcal{R}'_{t}) 
  &\leq \cost(\mathcal{R}^A_1, \ldots, \mathcal{R}^A_{i}) + \cost(\mathcal{R}'_{i+1}, \ldots, \mathcal{R}'_{t}) \\
  &\leq \cost(\mathcal{R}^A_1, \ldots, \mathcal{R}^A_{i}) + \cost(\mathcal{R}^B_{1}, \ldots, \mathcal{R}^B_{t}),
\end{align*}
where the second inequality is by monotonicity of the load function. 
\end{proof}

\end{document}